\documentclass{article}
\usepackage{geometry}
\usepackage{graphicx} 
\usepackage{comment}
\usepackage{amsmath}
\usepackage{amssymb}
\usepackage{braket}
\usepackage{amsthm}
\usepackage{svg}
\usepackage{tabularx}
\usepackage{graphicx}
\usepackage{subcaption}
\usepackage{adjustbox}
\usepackage{appendix}
\usepackage{float}
\usepackage{authblk} 
\usepackage[hidelinks]{hyperref}
\usepackage[colorinlistoftodos]{todonotes}
\usepackage{cleveref}
\usepackage{cite}

\DeclareMathOperator{\Tr}{Tr}
\newtheorem{theorem}{Theorem}

\title{Bridging Resource Theory and Quantum Key Distribution: Geometric Analysis and Statistical Testing}
\author{Andrea D'Urbano, Michael de Oliveira, Lu\'is Soares Barbosa}
\date{January 2024}

\author{
    \large{Andrea D'Urbano} \thanks{University of Salento, Lecce, Italy. Email: \texttt{andrea.durbano@unisalento.it}.} \
    \large{Michael de Oliveira} \thanks{International Iberian Nanotechnology Laboratory; LIP6, Sorbonne Universite; INESC TEC, Braga, Portugal. Email: \texttt{michael.oliveira@inl.int}.} \ 
    \large{Lu\'is Soares Barbosa} \thanks{University of Minho, Department of Computer Science, Braga, Portugal; INESC TEC, Braga, Portugal. Email:\texttt{lsb@di.uminho.pt}.}
    }

\begin{document}

\maketitle

\begin{abstract}

Discerning between quantum and classical correlations is of great importance. Bell polytopes are well established as a fundamental tool. In this paper, we extend this line of inquiry by applying resource theory within the context of Network scenarios, to a Quantum Key Distribution (QKD) protocol. To achieve this, we consider the causal structure $P3$ that can describe the protocol, and we aim to develop useful statistical tests to assess it.

More concretely, our objectives are twofold: firstly, to utilise the underlying causal structure of the QKD protocol to obtain a geometrical analysis of the resulting non-convex polytope, with a focus on the classical behaviours. Second, we devise a test within this framework to evaluate the distance between any two behaviours within the generated polytope. This approach offers a unique perspective, linking deviations from expected behaviour directly to the quality of the quantum resource or the residual nonclassicality in protocol execution.
\end{abstract}

\section{Introduction}

Since the inception of quantum mechanics, its portrayal as a non-local theory has markedly diverged from any classical theory. This divergence was later formalised by Bell's theorem, which provided a precise and feasible test to distinguish classical common-cause correlations from quantum ones. However, the task of testing for nonclassicality has remained non-trivial, as the realm of classical correlations and no-signalling correlations—those permitted by the theory of relativity—do not offer straightforward or efficient means to characterise the space of quantum correlations. Classical correlations can be represented as a convex polytope when studying the geometrical probability space of states and measurement operations, designated Bell polytopes \cite{Bell-nonlocality-review}, whereas the quantum set of correlations does not produce a polytope, lacking a simple and efficient description. Advancing our understanding in this area is of significant interest not only for foundational research but also for information processing, given that quantum correlations play a crucial role in their performance. Specifically, the capability of accurately distinguishing between quantum and classical correlations has implications in various applications, including device-independent quantum certifications \cite{PhysRevA.80.062327}, state discrimination processes \cite{bae2015quantum}, quantum key distribution \cite{xu2020secure}, proving computational advantage of quantum circuits over classical circuits \cite{Bravyi17,Watts19,coudron2021,Grier20,Daniel22,Oliveira22,Oliveira24} and other areas \cite{reviewStatisticalMethodsQuantumVerification}.

In this context, resource theories \cite{coecke2016mathematical} have emerged as a compelling theoretical approach.
 Within operational theories, the distinction between resourceful and resourceless processes (or specifically between classical and quantum processes \cite{Gilad-Chitambar2019}) is drawn based solely on the outcomes of operations, through a synthetic analysis of the process's inputs and outputs. Conversely, resource-theoretic approaches aim to identify the utilisation of fundamental resources, facilitating efficient descriptions and mappings between equivalent resources under free operations, which are typically defined as processes that do not increase the initial resource being analysed. While this shift moves away from the standard information-theoretic framework, potentially impacting its applicability in cryptographic contexts \cite{Gisin_2002}, it introduces a more robust descriptive capability for the specific causal structures under consideration and enhances the ability to certify nonclassicality \cite{PhysRevA.76.030305, PhysRevA.97.022111}.

Using the terminology of resource theory, this work is going to analyse scenarios similar to those encountered in typical Quantum Key Distribution (QKD) tasks \cite{scarani2009security}. 
Quantum Key Distribution holds significant interest both theoretically and experimentally due to its potential to provide secure communication protocols. Theoretical investigations continue to explore its foundations, including its reliance on quantum correlations as first explored by Bell experiments and the CHSH (Clauser-Horne-Shimony-Holt) inequality \cite{PhysRevLett.23.880, PhysicsPhysiqueFizika.1.195}. Experimentally, QKD research aims to implement and test these protocols in real-world scenarios, demonstrating their effectiveness in secure communication applications \cite{xu2020secure}. The study of QKD thus serves as a bridge between fundamental quantum principles and practical cryptographic systems, driving advancements in both realms.

Unlike prior studies that primarily focus on self-testing for certifying the security of QKD protocols \cite{10.5555/2011827.2011830, miklin2020semi, mehic2020quantum}, our analysis diverges by emphasising the quantification of quantum resources. We aim to use the language of resource theory and Bell Polytopes to perform a geometrical analysis. Additionally, we introduce statistical tests within this framework to quantify the proximity of the experiment to ideal conditions. This approach offers a novel perspective that links geometrical quantities to the quality of the quantum resource and nonclassicality in the protocol's execution.

\subsection{Outline}

After a general background (\Cref{sec:back}) describing Bell and networks scenarios and resource theories, we proceed in \Cref{sec:QKD-Bell-scenarios} with the description of how Bell scenarios and QKD protocols are related. In particular, using resource theory we describe Network scenarios representing QKD protocols and as an example we propose the well known BB84 protocol \cite{Bennett2014}, modifying it to better suit our resource theoretic description. More importantly we define the causal structure to be used, being a correlation scenario known as $P3$  \cite{fritz2012beyond}. 

In section \ref{sec_Geometry}, we proceed with a geometrical analysis of the resulting polytope. We note that the causal structure induces a breaking in the convexity, generating a non trivial geometrical shape. We describe the behaviour space, where this object lives, and extract explicit expression for the vertices (reported in appendix \ref{appendix:B}). We obtain a test to assess if 2 vertices are connected by a line lying inside the object or not, testing and characterising its convexity.

Finally in section \ref{sec:nonclassicality}, we develop a statistical test to assert weather the expected and observed behaviours are equal considering the errors (for example due to channel noise). We also provide a numerical example to give a clearer idea of the steps to follow. Furthermore we show that the distance in behaviour space is bounded by a sum of trace distances between the density matrices representing the expected and the observed system.

\section{Background}\label{sec:back}

In this section, we will present some essential background material on Bell scenarios, Network scenarios and the resource theories framework. Specifically, we will explore the Bell polytope, which serves as a foundational concept extensively utilised in this work. Additionally, we will provide an introduction to the methods of resource theories upon which the methods presented in this paper are built.

\subsection{Bell and Network scenarios} \label{sec_Bell_Scenarios}

A prevalent framework for investigating the distinctions between quantum and classical information theory is the "Bell scenario" \cite{Bell-nonlocality-review}. In this theoretical construct, a collection of $n$ parties jointly possesses a system that is partitioned into subsystems, one assigned to each party. At each individual location, participants can choose one of $m$ distinct measurements options, each capable of yielding $d$ outcomes. This configuration is represented as the ordered triplet $(n,m,d)$. In addition, these measurements are space-like separated, to preclude any causal influence among them.
Let's denote with $x_i \in \{1, 2, \cdots ,m\}$ the choice of measurement for the $i-$th party, also referred to as \textit{setting variables}, and with $a_i\in \{1,2, \cdots ,d\}$ the corresponding outcome, called \textit{outcome variables}. The following conditional probabilities,

\begin{equation}
    p(\textbf{a}|\textbf{x})=p(a_1, a_2, \cdots , a_n | x_1, x_2 , \cdots , x_n).
\end{equation}

map each possible combination of inputs and outputs variables in the Bell scenario to a point $\mathbf{Q} \in \mathbb{R}^{(md)^n}$ since there are $m^n$ measurements, each of which having $d^n$ possible outcomes. Following the conventions introduced in \cite{tsirelson1993}, the set of all these probabilities $\textbf{b} = \{p(\textbf{a}|\textbf{x})\}$ is referred to as \textit{behaviour} or, more informally, as \textit{correlations}.
The space where these points, which describe a behaviour in the Bell scenario, exist, must be subjected to some constraints and thus live in a smaller dimensional space. Indeed the normalisation constraint of the probabilities impose,

\begin{equation}
    \sum_{\textbf{a}} p(\textbf{a}| \textbf{x}) = 1 \qquad \forall \textbf{x}.
\end{equation}

\noindent This space must also be constrained by the non-negativity of the probabilities,

\begin{equation}
    p(\textbf{a}| \textbf{x}) \geq 0 \qquad \forall \textbf{a}, \textbf{x}.
\end{equation}

These inequalities defines a polytope $\mathcal{P}$ of dimension $m^n(d^n-1)$.
The presence of a specific physical model underlying the observed correlations in a Bell scenario imposes additional restrictions on the behaviours, reducing the dimension of the polytope.

A natural limitation stemming from relativity theory is the \textit{no-signalling} constraint \cite{cirel1980quantum, barrett2005popescu}, ensuring that no faster than light signals are transmitted,

\begin{equation}
    \sum_{a_i=1}^d p(\textbf{a} | x_1, \cdots , x_i,\cdots, x_n) = \sum_{a_i=1}^d p(\textbf{a} | x_1 , \cdots , x'_i,\cdots, x_n) \qquad \forall x_i \neq x'_i .
\end{equation}

The previous equation states that the marginal probabilities of each party are independent from the choice of measurement of that party and thus the choice of measurement of the $i-th$ party cannot be signalled to the other parties. 
The set of behaviours satisfying the no-signalling constrain is again a polytope, referred to as the \textit{no-signalling polytope} $\mathcal{P_{_{\mathcal{NS}}}}$, with dimension \cite{pironio2005lifting}: $[(d-1)m +1]^n -1$. 

A further constraint can be obtained considering the behaviours described by Local Hidden Variable (LHV) theories, where is taken into account the existence of concealed, deterministic variables inherent to each individual system, maintaining locality without invoking superluminal influences. Formally,

\begin{equation}
    p(\textbf{a}|\textbf{x}) = \int_\Lambda d \lambda q(\lambda) \prod_{i} p(a_i | x_i, \lambda)
\end{equation}

\noindent where $q(\lambda)$ is the probability density distribution of the hidden variable $\lambda$, defined over set $\Lambda$. This constraint implies that the probability distribution factorise over the parties, each influenced by a variable $\lambda$, which can be viewed operationally as a \textit{shared randomness}. 
Imposing the locality constraints, the no-signalling polytope, $\mathcal{P_{_{\mathcal{NS}}}}$ is further restricted to a subset: the \textit{local polytope} $\mathcal{P_{_{\mathcal{L}}}}$, having the same dimension of the no-signalling polytope \cite{tura2015nonlocality}. 

Another possible set of behaviours in a Bell scenario can be obtained by quantum mechanics. Formally, those behaviours corresponds to the points in $\mathcal{P}$ that can be written as,

\begin{equation}
    p(\textbf{a}|\textbf{x}) = \Tr \left( \rho \bigotimes_i M_{a_i|x_i} \right).
\end{equation}

Where $\rho$ is the density matrix describing the whole system in an Hilbert space $\mathcal{H} = \bigotimes_i \mathcal{H}_i$ of arbitrary dimension; $M_{a_i|x_i}$ are POVM measurement operators on the subspace $\mathcal{H}_i$ and without loss of generality \cite{tsirelson1993}
the system can always be considered pure, $\Tr(\rho^2)=1$, and the measurement can always be orthogonal projectors (if necessary by increasing the dimension of the Hilbert space). Therefore a quantum behaviour can always be written as,

\begin{equation}
    p(\textbf{a}|\textbf{x}) = \braket{\psi|M|\psi},
\end{equation}

\noindent where $M=\bigotimes_i M_{a_i|x_i}$ and the following conditions holds: $M_{a_i|x_i}M_{a'_i|x_i} =\delta_{a_i,a'_i}M_{a_i|x_i}$ ; $\sum_{a_i}M_{a_i|x_i}=\mathbb{I}_i$.
It can be shown \cite{pitowsky1986range}, that any local behaviour is also a quantum behaviour, thus $\mathcal{Q} \subset \mathcal{L}$, however, the converse does not hold.
The polytopal structure guarantees a minimal set of facet-defining inequalities, referred to as Bell inequalities, provided by the H-representation \cite{grunbaum1967convex} of the local polytope.
Within this framework, certain quantum behaviours exist that are not in the local polytope, forming a noteworthy set of behaviours that manifest correlations exceeding classical limitations, thereby violating Bell inequalities.
Additionally, any quantum behaviour satisfies the no-signalling constraints but the converse is not true \cite{pironio2005lifting}. 
In general, the set of quantum behaviours $\mathcal{Q}$ is not a polytope, but is convex and the following chain of inclusions holds: $\mathcal{L} \subset \mathcal{Q} \subset \mathcal{NS}$.
The dimensions of those sets is equal to the dimension of the no-signalling polytope \cite{pironio2005lifting},

\begin{equation} \label{eqn:dimensions_no_signalling_polytope}
    dim(\mathcal{P_{_{\mathcal{NS}}}}) =[(d - 1)m + 1]^n - 1.
\end{equation}

\subsection{Resource theory} \label{sec_Resource_Th}

To describe the transformations among behaviours in Bell scenarios (resources) under specific conditions (free operations) and explore the differences between classical and quantum correlations and how to use them in information theoretic settings, it can be used a framework called \textit{resource theory}  \cite{coecke2016mathematical}.
Many field of science use the notion of resource to describe processes and states (physical or logical). In this instance, the utilisation of a theoretical perspective on resources, enables the categorisation of non-classical behaviours within the Bell scenario \cite{wolfe2020quantifying} (see also \cite{schmid2020type, zjawin2023quantifying, hillery1999quantum, wagner2021using, duarte2018resource}). This approach establishes a framework that facilitates the description of no-signalling resources, non-local games, and the experiments designed to investigate them.

Let's now define a resource theory from intuitive considerations.
First and foremost, there may be different kind of resources, as well as the possibility to transform among them. This means that it should be introduced the set of the \textit{objects} resources and the set of the \textit{morphisms}, i.e. transformations among resources. Transformations should be designed in a manner that allows for sequential composition. In addition it should be possible to consider set of resources as a resource itself, and similarly it should be possible to compose in such manner also the transformations among resources acting in parallel on the sub-components of resources. Finally it is reasonable to assume the existence of an identity resource, that exerts no influence on other resources when combined with them. This heuristic facts can be formalised in the description of a resource theory as a symmetric monoidal category (SMC), where the objects are represents resources and the morphisms represent transformations among them implementable at no cost \cite{coecke2016mathematical}. 
SMCs offer an intuitive graphical calculus: there are theorems that establish a direct correspondence between equational reasoning within an SMC and the deformation of diagrams \cite{coecke2016mathematical}. 
It is important to observe that since the unit object is interpreted as the null resource, it inherently carries no associated cost. When considering the morphisms within the resource theory as transformations that incur no cost, it logically follows that any object that can be derived from the null resource also incurs no cost. The collection of all such objects is referred to as the \textit{free resources}, or \textit{free sub-theory}, while its complement constitute the \textit{costly} or \textit{non-free resources}.
The complete set of resources in some context is also referred to as \textit{enveloping theory}.

As in \cite{wolfe2020quantifying}, in what follows the set of free operation considered will be will be the set of all local operations assisted by shared randomness (LOSR) \cite{PhysRevLett.109.070401}. There are also other possible set of free operations, such as "wirings and prior-to-input classical communication"
(WPICC) \cite{PhysRevA.95.032118}.
The enveloping theory, on the other hand, will be the set of no-signalling behaviours, also called no-signalling boxes. 
In particular the resources, namely the Bell behaviours, are represented as probabilistic processes (seen as black-boxes) that link the input variables to the output variables, with the restriction that the inputs temporally precede the outputs, following a specific causal structure \cite{fritz2012beyond, branciard2012bilocal}.
There are powerful and general theories of causal structures, in particular the generalised causal models using generalised probability theories (GPT) \cite{janotta2014generalized}. However in the following, it is just needed the notion of \textit{causal structure}, seen as a directed acyclic graph where the nodes represents variables and the edges direct causal connections among them. 
Using the graphical representation of resources and causal structures those situations can be explained in detail. 
To understand how this graphical representation of resources and  causal structure works, consider figure \ref{fig:bipartite_schema}, where two parties are influenced by some common cause. If the specific box can be described by a classical causal model (i.e. it is a local box, whose conditional probability distribution respect Bell inequalities), the common cause is classical and can be represented by a shared randomness. 
The boxes represents probabilistic process from the input wires to the output wires. Single lines describe classical correlations, while double lines describe non-classical (in this case quantum) systems. 
Those classically realisable boxes (i.e. realisable using only LOSR operations) represent the free sub-theory and are described by local Bell behaviours, while the costly common-cause boxes are the no-signalling boxes violating some Bell inequalities.

\section{QKD Bell scenarios} \label{sec:QKD-Bell-scenarios}

Building upon the resource theory foundation, we now shift our attention to applying these concepts to bipartite Quantum Key Distribution (QKD) scenarios. We will explore how resource theory and Bell scenarios connect. We will then add an eavesdropper to the picture and discuss the implications taking into account the resulting non trivial causal structure.

\subsection{Bipartite Bell scenario}

The diagram in figure \ref{fig:bipartite_schema}, is divided into two wings: the left side corresponds to Ashley, the second to Charlie. Causal correlation flows from bottom to top and the boxes at the centre represent a common cause to the wings. 

Let's now examine a general LOSR transformation applied to the resource just described. This general transformation is depicted in the blue section of the diagram: the common cause box represent the shared randomness $\lambda$, generated with probability distribution $q(\lambda)$; the other boxes represent the pre and post processing of Ashley and Charlie variables.

Having discussed each element of the diagram, its meaning is now clearer: a resource (the left orange diagram) is transformed using LOSR operations (the outside blue diagram) into another resource.

Wolfe et al. in \cite{wolfe2020quantifying} proved the following theorem.

\begin{theorem} \label{theorem:convexity} 
The set LOSR transformation for the bipartite Bell scenario with a common cause (diagram \ref{fig:bipartite_schema}) is convex, i.e., if $\tau_0,\tau_1 \in LOSR$, then $\omega \tau_0 + (1-\omega) \tau_1 \in LOSR$ for $0 \leq \omega \leq 1$.
\end{theorem}

In addition, it can be showed that the set of convexly extremal transformations are deterministic operations. This proposition represents a slight extension of Fine's argument \cite{fine1982hidden}.
The consequence of this theorem is that the structure of free operations is a \textit{Polytope}, and consequently any local behaviour (i.e. any behaviour inside the polytope) can be described by a convex combination of deterministic operations (i.e. the vertices). Therefore the subsequent analysis will use the vertices of the polytope to generate all accessible behaviours.

\begin{figure}[ht]
    \centering
    \includegraphics[width=10cm]{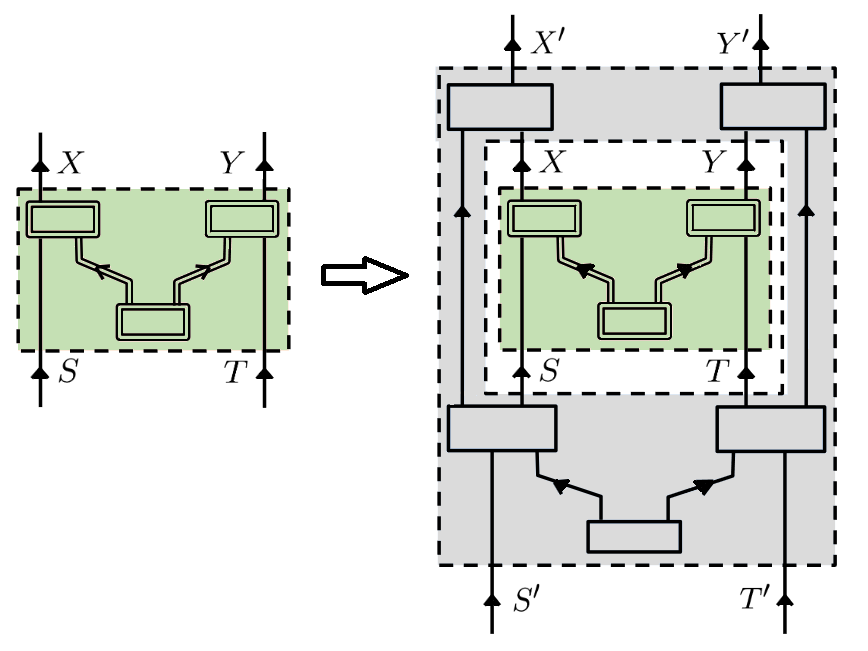}
    \caption{The diagram describe a generic bipartite LOSR transformation, with a direct correspondence with Bell Experiments. The left diagram, described as a conditioned probability $P_{XY|ST}$, is transformed using LOSR operations in the right diagram $P_{X'Y'|S'T'}$.}
    \label{fig:bipartite_schema}
\end{figure}

\subsection{Network scenario}

The use of non trivial causal structures is reflected into the structure of the resource theory: specifically, the convexity (as in theorem \ref{theorem:convexity}) no longer holds in general scenarios.
Such more general scenarios, involving various independent sources and parties arranged in a network, are referred to as Network scenarios \cite{Tavakoli.2022}.

In what follows it will be taken into consideration a particular causal structure, depicted in figure \ref{fig:P3_correlation}, described by the path graph $P_3$ and thus called $P_3$ correlation scenario \cite{fritz2012beyond}, or three-on-line correlation scenario \cite{navascues2020inflation}.

\begin{figure}[ht]
    \centering
    \includegraphics[width=4cm]{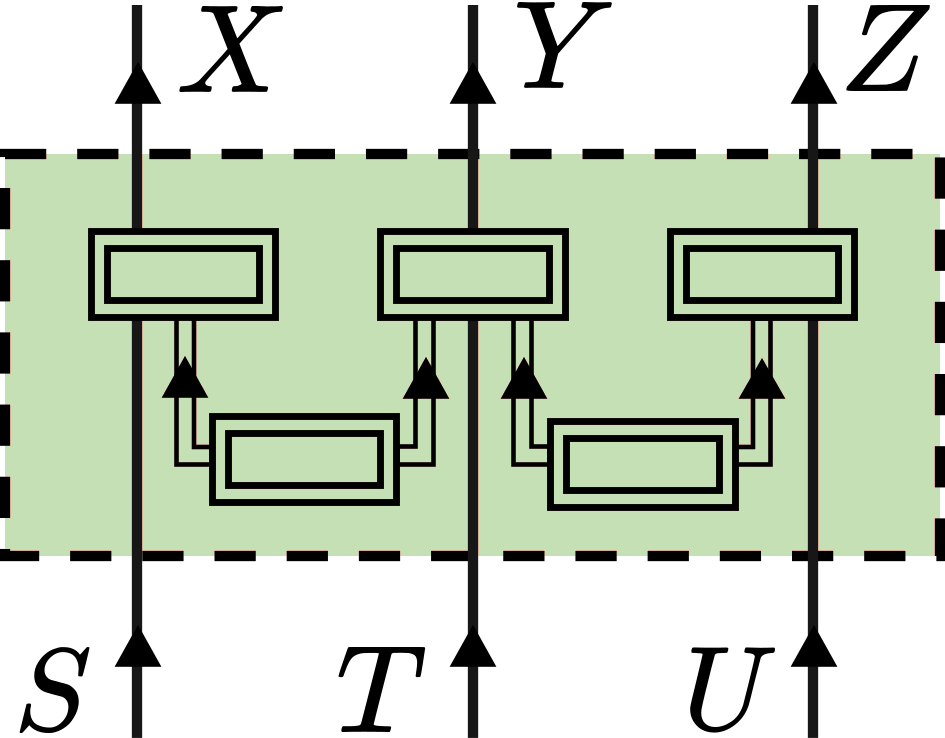}
    \caption{Scenario under consideration: it is a network scenario (3,2,2), three parties and dichotomic setting and output variables. The causal structure induce a non-convexity in the geometrical description of resources. The two common cause boxes are general, i.e. quantum systems, in a classical realisable situation those represent a classical shared randomness among adjacent couples of parties.}
    \label{fig:P3_correlation}
\end{figure}

This particular causal structure describe a Network nonlocality scenario, where three parties have dichotomic setting variables ($S, T, U$) and dichotomic output variables ($X,Y,Z$). The causal structure is described by the path graph $P_3$, which consist of three nodes in a line, connected by 2 edges: the nodes represents the parties, while edges represents common causes among them. In a general scenario the common cause is non classical, but in a classical realisable scenario, the common cause is represented by shared randomness. 
This non trivial causal structure impose a non convexity on the geometrical structure of resources (using LOSR as free operations), viewed as behaviours
using the terminology developed in Bell experiments. 

The framework developed to describe the $P_3$ scenario can be applied to real cases of a quantum communication protocols.
The idea is to gain new insight into known techniques using a novel description. 
Using this scenario we can represent the BB84 protocol \cite{Bennett2014, shor2000simple}, with a slight modification to make it fit into the framework.
It is worth noting that, even if We are not interested in self testing, the description can be seen as a prepare and measure scenario \cite{tavakoli2018self}.

In particular the description of the communication steps is essential in order for BB84 to adhere to the restriction of the scheme analysed.  
As in every protocol of quantum key distribution (QKD), two parties, called Ashley and Charlie, want to obtain a shared secret bit-string through quantum means. In our setting, we must note that LOSR already provides shared random bit-string. However, the secrecy and the randomness of the shared bit-string are not guarantee. During the QKD protocol, the entanglement is converted to a secret and random bit-string
\footnote{Note that LOSR already provides shared-random bit-strings. However, to assure the security of a communication protocol, the randomness and exclusive knowledge of the bit-string must be guaranteed. QKD achieves these two properties elegantly through the postulates of quantum mechanics. Additionally, the shared randomness considered in the resource theory does not have these properties defined or guaranteed.}. 
In this modification of BB84, Ashley create a batch of entangled pair of qubits, for example Bell pairs $|\Phi^+\rangle =\frac{1}{\sqrt{2}}(|00\rangle+|11\rangle)$, and send through a quantum channel one qubit of each pair to Charlie. After this first step, all communication channels must be severed for the measurement step. In the second phase of the protocol, Ashley and Charlie each independently and randomly select a bit ($s_i$ and $u_i$) for every entangled pair $i$. These bits represent the chosen measurement bases for the qubits, between the standard base $\{|0\rangle,|1\rangle\}$ and the Hadamard base $\{|+\rangle,|-\rangle\}$. If both Ashley and Charlie measure the same pair in the standard base there is an equal probability of measuring 0 or 1 and both will obtain the same result. If the measurement is performed in the Hadamard base the situation is exactly the same,
\begin{align}
    |\Phi^+\rangle &=\frac{1}{\sqrt{2}}(|00\rangle+|11\rangle) \nonumber \\ &=\frac{1}{2\sqrt{2}}\Big( (|+\rangle+|-\rangle)(|+\rangle+|-\rangle) + (|+\rangle-|-\rangle)(|+\rangle-|-\rangle) \Big)  \nonumber \\
    & = \frac{1}{\sqrt{2}}(|++\rangle+|--\rangle).
\end{align}

On the other hand if Ashley chooses a measurement base different from Charlie, the results will be completely uncorrelated,
\begin{equation}
    \langle \Phi^+ | X_AZ_C | \Phi^+ \rangle = \langle \Phi^+ | Z_AX_C | \Phi^+ \rangle = 0.
\end{equation}

In the last step of the protocol, after the measurements have been performed, Ashley and Charlie announce publicly the bases used $s_i$ and $u_i$. Only the measure performed casually in the same base will be used to generate a secret shared key, while the others will be used to verify the its security.

The crucial phase of the protocol, the measurement step, can be described by an ensemble of bipartite Bell scenarios with the same causal structure as in diagram \ref{fig:bipartite_schema}.
Using the terminology of Bell experiments, $s_i$ and $u_i$ are the setting variables, while the measurements (dichotomic) are the output variables.

Consider what happens if another character, called Blake, wants to hijack the protocol. A strategy that Blake could follow is a men in the middle attack at the step of quantum communication. He could interpose between the parties and receive the qubit sent by Ashley, sharing a Bell pair with her, and then send a qubit of a Bell pair to Charlie, pretending to be Ashley and sharing another Bell pair with Charlie. This situation is described by the P3 correlation scenario depicted in figure \ref{fig:P3_correlation}.

The language of resource theory and Bell (Network) scenario can be used to describe protocols and scenarios similar to the one just presented. 
For example, the use of the state $| \Phi^+ \rangle$ as quantum common cause produces: $a_1b_1= 1/2$; $a_2b_2= 1/2$; $b_1c_1= 1/2$ and $b_2c_2= 1/2$.
The meaning of the previous qualities is the following: if parties sharing a common cause (A and B or B and C), measure in the same base, the probability to observe something equal is 1, but the precise result is not known. 
Those constraints should actually be verified and not manually imposed, however it can be useful to consider them in theory.

Using the diagram \ref{fig:P3_correlation} for the description of the general resource in this scenario and diagram \ref{fig:P3_transformation} to express the general LOSR transformation, the non convexity resulting from this choice of causal structure can be proved.

The diagrammatic description gives rise to:

\begin{align} \label{eqn:p3_resources}
    P_{XYZ|STU}(xyz|stu)=& \nonumber \sum_{\lambda_A,\lambda_B,\lambda_C} P_{X|S\Lambda_A}(x|s \lambda_A)P_{Y|T\Lambda_B}(y|t \lambda_B)P_{Z|U\Lambda_C}(z|u \lambda_C) \\ \nonumber & \times P_{\Lambda_A,\Lambda_{B_1}}(\lambda_A,\lambda_{B_1}) P_{\Lambda_{B_2},\Lambda_C}(\lambda_{B_2},\lambda_C) \\
    = & \sum_{\lambda,\lambda '} P_{X|S \Lambda}(x|s \lambda) P_{Y|T \Lambda \Lambda'}(y|t\lambda\lambda') P_{Z|U\Lambda'} (z|u\lambda') P_\Lambda(\lambda) P_{\Lambda '}(\lambda ').
\end{align}

In the final step, it was used $\Lambda_B=(\Lambda_{B_1},\Lambda_{B_2})$ and $\Lambda_A \equiv \Lambda_{B_1}$, $\Lambda_{B_1} \equiv \Lambda_C$. The same structure of the previous equation naturally arise also in the description of the LOSR transformation using the same causal structure. In particular, using the diagram \ref{fig:P3_transformation}, the general transformation, i.e. free operation, is found to be 

\begin{align} \label{eqn:P3_transformations}
    P_{X'Y'Z'STU|XYZS'T'U'}(x'y'z'stu|xyzs't'u')= \nonumber \\
    =\sum_{\lambda,\lambda'} P_{X'S|XS'\Lambda}(x's|xs'\lambda) P_{Y'T|YT'\Lambda \Lambda'} (y't|yt'\lambda \lambda') P_{Z'U|ZU'\Lambda'} (z'u|zu'\lambda') P_\Lambda(\lambda) P_{\Lambda '}(\lambda '),
\end{align}

\noindent this set is not convex: this is caused from the presence of two sources of randomness. Following the steps of the proof \ref{theorem:convexity}, it is not possible anymore to propagate along both external parties the same randomness to influence the choice of a strategy. In particular $A$ and $B$ are not able to use the same shared randomness: this is a crucial observation, it is actually the fundamental reason for the convexity breaking, and it will be useful to characterise geometrically the resulting polytope.

\begin{figure}[ht]
    \centering
    \includegraphics[width=6cm]{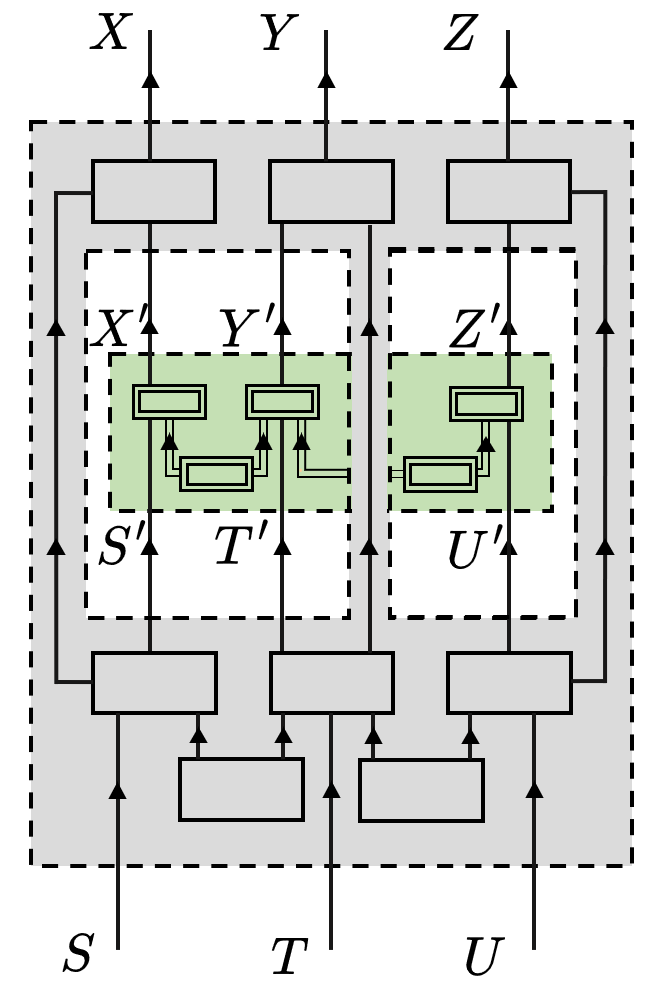}
    \caption{LOSR transformation of the P3 causal structure scenario. In green, a resource with a P3 causal structure, with 3 classical input wires and 3 classical output wires. In grey, the most general LOSR transformation following the same causal structure.}
    \label{fig:P3_transformation}
\end{figure}

\section{Geometrical analysis} \label{sec_Geometry}

\subsection{Geometric description of the QKD protocol}

The main objective of this section is to characterise geometrically the polytope describing the free resource theory of the scenario under study.

The set of resources must not be confused with the set of free transformations. As noted in \cite{wolfe2020quantifying}, the set described in equation \ref{eqn:P3_transformations}, can be transformed in a form equal to equation \ref{eqn:p3_resources} by the change of variable: $\tilde{X}=(X',S)$, $\tilde{S}=(X,S')$, and so on.
However, our focus is to study the free resources and not the transformations among them. Since one of the goals of this work is to find a link between this abstract description and a useful protocol, it will be studied the polytope describing the free resources, existing in the behaviour space of the Bell scenario. It will be possible to extract some information on the transformation among resources analysing the lines (convex transformations) connecting points (resources).

As shown in \cite{fine1982hidden},
the extremal points of the local polytope describe deterministic strategies and therefore enumerating them is a relevant question. A general local behaviour is a linear combination of some deterministic strategies, but the breaking of convexity caused by the causal structure imply that not every linear combination is allowed.  
To enumerate all the deterministic strategies it will be used a technique of tensor decomposition.
The measurement choice of each party is dictated by the actual values $(s,t,u)$ of the setting variables $(S,T,U)$. To enumerate all the possible deterministic local strategies, it is possible to decompose the arbitrary strategy in the following manner,

\begin{equation} \label{eqn:tensor_decomposition}
    D= \textbf{e}_ {S_ 0,S_ 1}\otimes \textbf{e}_ {T_ 0,T_ 1}\otimes \textbf{e}_{U_ 0,U_ 1}.
\end{equation}

The 4 dimensional real vectors $\textbf{e}$ describe a behaviour in the corresponding wing. Each $\textbf{e}_{i,j}$ is an element of the canonical base of $\mathbb{R}^4$; the $i$-th canonical base describe the $i$-th possible deterministic local strategy ordered by concatenating the expected output variables. Taking for example the strategies for the first wing,

\begin{equation}
    \begin{cases}
        \textbf{e}_{0,0}=(1,0,0,0) \qquad \implies \qquad S=0\mapsto X=0, \ S=1 \mapsto X=0 \\
        \textbf{e}_{0,1}=(0,1,0,0) \qquad \implies \qquad S=0\mapsto X=0, \ S=1 \mapsto X=1 \\
        \textbf{e}_{1,0}=(0,0,1,0) \qquad \implies \qquad S=0\mapsto X=1, \ S=1 \mapsto X=0 \\
        \textbf{e}_{1,1}=(0,0,0,1) \qquad \implies \qquad S=0\mapsto X=1, \ S=1 \mapsto X=1
    \end{cases}.
\end{equation}
    
The vectors $\textbf{e}$ are the bases for the general strategies described by $\textbf{r}$. 
A corresponding notation, closer to the Bell experiment description, is in terms of conditional probabilities to observe a fixed result (for example a specific direction, say "up", for a spin measurement). For example $P(X=0|S=0)$ is shortened to $P(X|0)$. For the sake of convenience these conditional probabilities will be referred to as,

\begin{equation} \label{eqn:conditional_probabilities}
    \begin{cases}
        P(X|0)=a_0 \quad , \quad P(X|1)=a_1 \\
        P(Y|0)=b_0 \quad , \quad P(Y|1)=b_1 \\
        P(Z|0)=c_0 \quad , \quad P(Z|1)=c_1 
    \end{cases}.
\end{equation}

In total there are 64 of such deterministic strategies, this is directly deduced from the tensor decomposition \ref{eqn:tensor_decomposition}: there are 3 terms, each with 4 possible choices, and thus a total of $4^3=64$ strategies.
These strategies can be represented directly as base vectors in $\mathbb{R}^{64}$ (from the decomposition \ref{eqn:tensor_decomposition}), or can be represented by restricting the conditional probabilities in equation \ref{eqn:conditional_probabilities} to represent only deterministic behaviours. In this second case, the $64$ strategies are represented by $6$ dimensional vectors of the form

\begin{equation} \label{eqn:det_strategies}
    D=(a_0,a_1,b_0,b_1,c_0,c_1),
\end{equation}

\noindent where each entry, representing a deterministic conditional probability, can be either 0 or 1. To completely characterise the scenario under study more degree of freedom, in particular 26, are needed as prescribed by equation \ref{eqn:dimensions_no_signalling_polytope}. Different parametrizations exist, for example using correlators \cite{bancal2010looking}, but since it was already introduced the vector $D$ composed by the 6 conditional probabilities of the single measurement events, it is straightforward to expand its definition to include all possible composite events.
For convenience events of the form $P(XY|00)$ can be expressed as $a_0 b_0$ and similarly for the other combinations (since the logical AND on the events is mapped to the multiplication on the probabilities). In total, in addition to the 6 single events, there are 12 possible composite events of pairs,

\begin{equation}
    Pairs=(a_0b_0, a_0b_1, a_1b_0, a_1b_1, a_0c_0, a_0c_1, a_1c_0, a_1c_1, b_0c_0, b_0c_1, b_1c_0, b_1c_1)
\end{equation}

and 8 possible events of triplets,

\begin{equation}
    Tris=(a_0b_0c_0, a_0b_0c_1, a_0b_1c_0, a_0b_1c_1, a_1b_0c_0, a_1b_0c_1, a_1b_1c_0, a_1b_1c_1)
\end{equation}
adding up to 26 conditional probabilities as expected. Since there are only 6 free variables, there is a total of 64 of such vectors, representing an extremal point of the polytope under study. Each vertex has the form

\begin{equation} \label{eqn:vertex_structure}
    V=(D, Pairs, Tris),
\end{equation}

\noindent and a list of all the vertices is found in appendix \ref{appendix:A}.
To obtain all possible local strategies, is not possible to find the convex hull of the vertices just found, because the causal structure imply a breaking of the convexity. To characterise this geometrical object let's start by describing where the vertices lies. Since the entries of each vertex can be either 0 or 1, the set of 64 vertices is a subset of the nodes within the hypercube graph $Q_{26}$, obtained considering the edges and nodes of a 26D hypercube. Furthermore, the first 6 entries of the vertices (containing the only free variables) describe a 6D hypercube, that can be embedded into the 26D hypercube by setting all the other 20 entries to 0. This means that the full polytope can be projected down to a 6D hypercube by absorbing the 20 dimensions given by the conditional probabilities of composite events, or conversely, the polytope under study can be created by relocating each vertex of the 6D hypercube to a node within the corresponding 20 dimensional hypercube graph $Q_{20}$ (all 64 $Q_{20}$ graphs live in parallel spaces) expanding from that vertex. 

A first exploratory analysis of this object can be made by clustering the vertices by their hamming weight (the number of entries equal to 1). This division into classes, found in figure \ref{fig:hamming_weight}, amounts to counting vertices in 25D spaces orthogonal to the main diagonal of the 26D hypercube, from vertex $\vec{0}$ to $\vec{1}$. 

Using tensor decomposition of strategies and the steps of the proof for non convexity, we obtain a rule to test if two vertices of the polytope are connected by a line lying inside the polytope or not; in other words we test if they see each other. Borrowing from the computer science terminology of the visibility graph \cite{WELZL1985167}, two vertices will be visible of hidden if there exist a line connecting them lying inside the geometrical shape.

\begin{theorem} \label{theorem:test}
    Given the causal structure P3,  
    let $D= \textbf{e}_ {S_ 0,S_ 1}\otimes \textbf{e}_ {T_ 0,T_ 1}\otimes \textbf{e}_{U_ 0,U_ 1}$ and $D'= \textbf{e}'_ {S_ 0,S_ 1}\otimes \textbf{e}'_ {T_ 0,T_ 1}\otimes \textbf{e}'_{U_ 0,U_ 1}$ be two deterministic local strategies, representing the extremal points of the corresponding local polytope, then the corresponding points in the local polytope are connected by a line lying inside the polytope if and only if $\textbf{e}_ {S_ 0,S_ 1}=\textbf{e}'_ {S_ 0,S_ 1}$ or $\textbf{e}_{U_ 0,U_ 1} =\textbf{e}'_{U_ 0,U_ 1}$.
\end{theorem}

\begin{proof}
    Using theorem \ref{theorem:convexity}, the line we refer to can be described as a convex combination of 2 deterministic strategies. Let's consider $\beta_A$ as a binary variable responsible for determining whether the first party, will execute $\textbf{e}_ {S_ 0,S_ 1}$ or $\textbf{e}'_ {S_ 0,S_ 1}$ and, similarly, $\beta_C$ the binary variable responsible for the third party decision. The variable $\beta_A$ is sampled from a distribution $p_{\beta_A}$, where $p_{\beta_A} (0)=\omega_A$ and $p_\beta (1)=1-\omega_A$ and similarly for the third party. $\beta_A$ and $\beta_C$ are shared respectively between the first and the middle party, and the third and the middle party. The local processes executed by each party are dependent upon $\beta_A$ and $\beta_A$, in particular the first party depend only on $\beta_A$, the third party only on $\beta_C$ and the middle party on both. This means that if $\beta_A = 0$, then $\textbf{e}_ {S_ 0,S_ 1}$ is implemented by the first party (similarly for $\beta_A = 1$ and for $\beta_C$). Since there is no shared resource between the first and the third party, $D$ can be transformed trough a convex combination in $D'$, with parameter $\omega_A$ or $\omega_C$, if and only if  $\textbf{e}_ {S_ 0,S_ 1}=\textbf{e}'_ {S_ 0,S_ 1}$ or $\textbf{e}_{U_ 0,U_ 1} =\textbf{e}'_{U_ 0,U_ 1}$.
\end{proof}

\begin{table} 
    \centering  
    \begin{tabular}{|ccc|cc|}
    \hline
    $\textbf{e}_ {S_ 0,S_ 1}$ & $\textbf{e}_ {T_ 0,T_ 1}$ & $\textbf{e}_{U_ 0,U_ 1}$ & \# vertices & status \\ \hline
    $\textbf{e}_ {S_ 0,S_ 1}$ & $\textbf{e}_ {T_ 0,T_ 1}$ & $\textbf{e}_{U_ 0,U_ 1}$ & 1 & coincident \\
    $\textbf{e}_ {S_ 0,S_ 1}$ & $\textbf{e}_ {T_ 0,T_ 1}$ & $x$ & 3 &  visible \\
    $\textbf{e}_ {S_ 0,S_ 1}$ & $x$ & $\textbf{e}_{U_ 0,U_ 1}$ & 3 &  visible \\
    $x$ & $\textbf{e}_ {T_ 0,T_ 1}$ & $\textbf{e}_{U_ 0,U_ 1}$ & 3 & visible \\
    $\textbf{e}_ {S_ 0,S_ 1}$ & $x$ & $x$ & 9 & visible \\
    $x$ & $x$ & $\textbf{e}_{U_ 0,U_ 1}$ & 9 & visible \\
    $x$ & $\textbf{e}_ {T_ 0,T_ 1}$ & $x$ & 9 & hidden \\
    $x$ & $x$ & $x$ & 27 & hidden \\ \hline 
    \end{tabular}
    \caption{From each vertex a total of 27 vertices can be seen, and other 36 are hidden.}
    \label{tab:test}
\end{table}

Table \ref{tab:test} shows how the test described in theorem \ref{theorem:test} works. Specifically, it allows to determine the count of visible vertices and those concealed from an individual vertex. 
To gain a complete understanding of this geometrical shape in principle it could be possible to group the vertices into convex shapes (using the test) and then find their union. 
To find those convex shapes whose union results in the full non convex polytope, the test can be used to find subsets of the vertices, such that all the vertices in the same subset (i.e. convex shape) can "see" all others. However this purely geometrical problem is out of the scope of this paper.
Two related interesting questions, defined considering the graph $\mathcal{G}=(n,e)$, where $n$ is the set of nodes defined by the set of vertices of the polytope under study, and $e$ is the set of edges, defined by the pairs of vertices resulting in a positive test as defined in theorem \ref{theorem:test}, are:
\begin{itemize}
    \item the identification of the minimum number of edges to traverse needed to travel from any vertex to any other;
    \item the minimum vertex cover or the minimum edge cover of the visibility graph, obtained using the test in theorem \ref{theorem:test};
\end{itemize}

Let's start from the first problem, related to the more famous All Pairs Shortest Path (APSP) problem \cite{cormen2022introduction}. In APSP, the objective is to find the shortest path from every possible node to any other. In this case however, we are interested in the maximum length among all those possible shortest paths, since intuitively this information can give an idea of the roughness (amount of "nooks and cranny") of the shape.
Its solution is found in the following theorem.

\begin{theorem} \label{theorem:minimum_hops}
    Given the graph $\mathcal{G}=(n,e)$, where $n$ is the set of nodes defined by the set of vertices of the polytope under study, and $e$ is the set of edges, defined by the pairs of vertices resulting in a positive test as defined in theorem \ref{theorem:test}; the maximum of the all pair shortest path is 2.
\end{theorem}

\begin{proof}
    The proof can be deduced by table \ref{tab:test}. In particular, given an arbitrary "origin" vertex, all the 27 visible vertices are connected by an edge by definition of $\mathcal{G}$ and therefore their shortest path length is 1. In addition, every other "destination" vertex among the 36 not incident with the origin can be reached with an additional edge, using a path passing through a visible vertex from the origin sharing with the destination at least one component among $\textbf{e}_{U_ 0,U_ 1}$ and $\textbf{e}_{S_ 0,S_ 1}$. 
\end{proof}

A similar approach can be used for the second question, related to the minimum edge cover problem. This famous mathematical problem is defined as the problem of finding a set containing the minimum possible number of edges such that all vertices are incident to at least one edge in the set. Of course standard algorithms to solve minimum edge cover and minimum vertex cover can be used, but it is interesting to note that an alternative route can be followed. 
The minimum set of vertices, that will be referred to as \textit{generators} from which all other vertices are in sight can be found. 
Since from each vertex there are 27 other vertices reachable, a lower bound on the number of generating vertices needed is 
$\lceil 64 / 28 \rceil=3$. The actual number of generators is 4, as showed in the following theorem.
We define the \textit{visibility graph} of a geometrical shape as the graph having as nodes the vertices of the shape, connected by an edge if and only if the two corresponding vertices in the geometrical shape are connected by a line lying inside it.

\begin{theorem} \label{theorem:vertices_generators}
    The minimum number of generators for the visibility graph of the local strategies in the P3 causal structure is 4.
\end{theorem}

\begin{proof}
    The proof is constructive and it gives two different ways to find the generators. The test shows that each vertex can be linked to 27 other vertices by lines lying inside the full polytope under study.  Considering that the number of all vertices is 64, a lower bound on the number of generators is 3 (since the ceiling of $\lceil 64/28 \rceil=3$). 

    As a second step, this lower bound is improved to 4. 
    For the sake of brevity, the deterministic strategies in equation \ref{eqn:det_strategies} will be denoted as a triplet $(\alpha_i,\beta_j,\gamma_k)$ where $i,j,k \in \{0,1,2,3\}$. Since the test in theorem \ref{theorem:test} amounts to consider if $i$ or $k$ change, independently from $j$, the vertices can be divided into 16 classes $(\alpha_i,\gamma_k)$, representing a 4 by 4 matrix. Given a vertex, and therefore an element of this matrix, the test marks as visible the vertices in the same row and in the same column of the generator. Therefore, it is impossible to cover all the elements of this matrix with less than 4 generators and the lower bound on the minimum number of generators is 4.
    This is also the exact number, since there are explicit construction involving 4 generators, considering for example all the diagonal elements $(\alpha_i,\gamma_i)$ or just a single column or row. In appendix \ref{appendix:B} can be found those explicit constructions.
\end{proof}

\begin{figure}
    \centering
    \includegraphics[width=10cm]{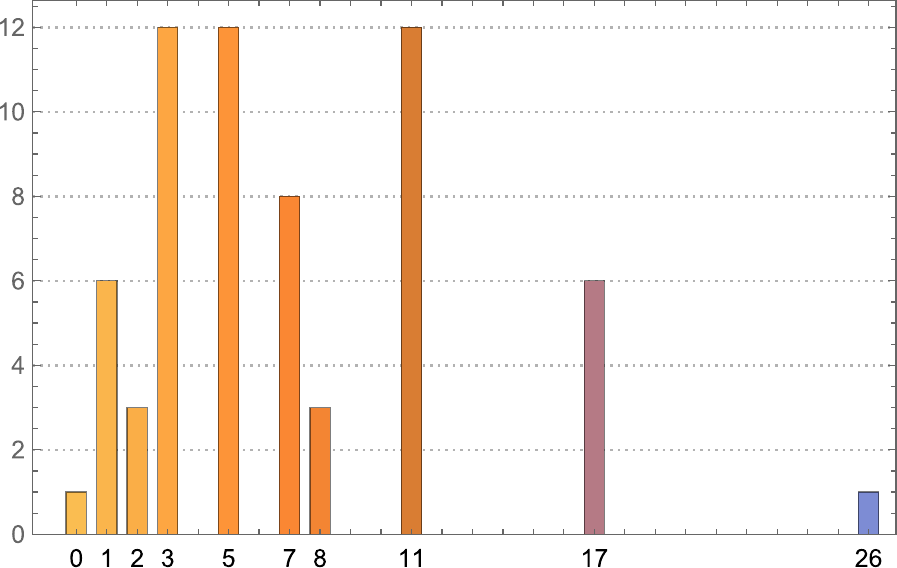}
    \caption{Clustering of vertices by their hamming weight. This division into classes, amounts to counting the number of vertices in spaces orthogonal to the 26D hypercube main diagonal.}
    \label{fig:hamming_weight}
\end{figure}

\section{Extracting nonclassicality} \label{sec:nonclassicality}

Translating resource theoretical descriptions into practical protocols requires understanding the information obtainable from actual physical measurements. A key point is that any specific implementation of a Bell experiment, or more in general Network Scenario, corresponds to a single point in the behaviour space. Ideally, experimental data would allow us to pinpoint the exact location of this point within the space, fully characterising the experiment.

In reality, experimental errors are unavoidable. Factors like imperfect measurement devices, environmental noise, and limitations in preparing entangled states contribute to deviations from the ideal scenario.  Consequently, instead of a single point, the experimental data defines a region in the behaviour space. The size of this region reflects the magnitude of the combined errors. The coordinates of the centroid of this region are determined by the measured correlations between the different measurement outcomes. These correlations are calculated by averaging over many repetitions of the experiment, using an ensemble of identically prepared systems (e.g., entangled particle pairs).

Consider a scenario where two parties, Ashley and Charlie, each perform measurements on their respective parts of the entangled system. Each time they choose a specific measurement setting, they effectively fix half of the variables defining a point in the behaviour space.  Therefore, a single measurement run provides information about an entire hyperplane within the space. By repeating the experiment with different measurement settings, we obtain multiple such hyperplanes, each with an associated uncertainty due to experimental error.  The final representation of the experiment in the behaviour space is given by the intersection of these hyperplanes and their corresponding error regions.

The goal is to be sure that the point found is outside the non convex local polytope described previously, proving that no classically realisable boxes could have generated the observed behaviour.
What just stated is only partially correct, since no one will ever have access to Blake's measurements or strategies. Therefore the same considerations must be projected to a smaller space obtained by marginalising Blake's variables.
The analogue of equation \ref{eqn:vertex_structure}, defining the vertices structure is,

\begin{equation} \label{eqn:8dvertex_structure}
    V = (a_0,a_1,c_0,c_1,a_0c_0,a_0c_1,a_1c_0,a_1c_1) .
\end{equation}

The explicit list of all $2^4=16$ vertices is found in appendix \ref{appendix:A}. This dimensionality reduction causes a degeneration of groups of four vertices in 26 dimensions into one in 8, corresponding to all deterministic strategies that differ only for the choice of Blake (i.e. 4 possible deterministic strategies).
The convexity test, to check if two vertices can see each other or not, is directly translated in the reduction. 
The same construction can be followed to reach similar geometrical conclusions for the reduced polytope in 8 dimensions, as reported in appendix \ref{appendix:B}.

Apart from all the geometrical reasoning, interesting per se, it makes sense to search for a method to exploit this geometrical interpretation to identify any potential protocol tampering resulting from eavesdropping or noise interference.
An useful observation is that equation \ref{eqn:8dvertex_structure} actually describe a 4 dimensional manifold passing for all vertices. Every point on the manifold describe a complete absence of correlations since, the probability of the composite events are just the product of probabilities of single events. The further a configuration is from this surface, the more complex the underlying correlations are. It is also easy to compute the expected point describing the correlations given the exact resources shared. For example, consider as before the Bell state  $| \Phi^+ \rangle$, acting as a common cause between Ashley and Charlie. They do not know of the presence of Blake, therefore they believe to be in the configuration depicted in figure \ref{fig:bipartite_schema}.
In that case the variables describing the single events would be $a_0=a_1=c_0=c_1=1/2$ while the only composite events to change would be $a_0c_0=1/2$ and $a_1c_2=1/2$ due to entanglement. The corresponding point on the uncorrelated manifold would have all composite events equals to $a_0c_0=a_0c_1=a_1c_0=a_1c_1=1/4$. It can be found the minimum distance of the expected point from the uncorrelated surface (using numerical methods or Lagrange multipliers) to gauge the effect of noise. In this particular case, the minimum distance from the point to the surface is about $0.92$ and the if the single event probabilities  are all equal approximately to $0.564$, it is found the closest point on the surface.

Using this distance as normalisation, it can be gauged the effect of noise. 
Moreover, a hypothesis test can be conducted to ascertain the level of confidence at which one can assert that the measured point and the expected point are indeed the same.

\begin{figure}
    \centering
    \includegraphics[width=\textwidth]{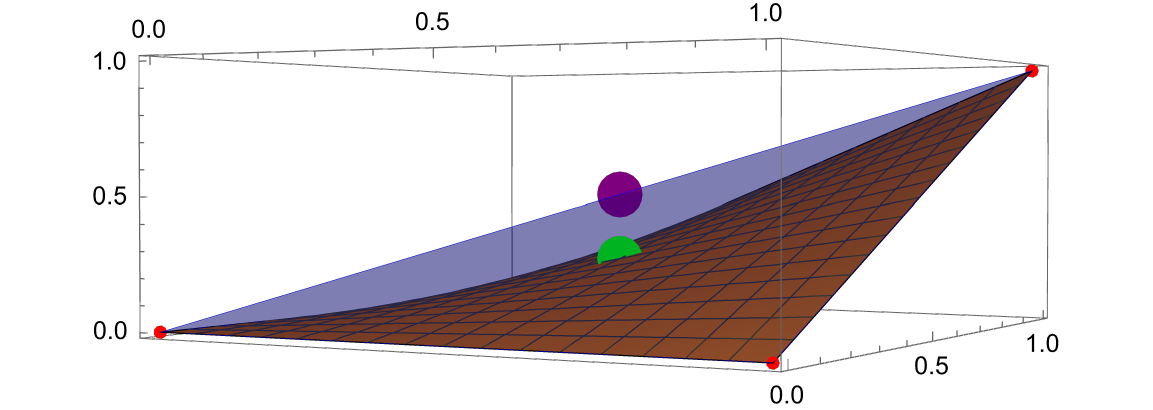}
    \caption{Example in low dimensionality of the techniques that can be employed to describe the protocol. In red are indicated local deterministic strategies, in orange the "uncorrelated manifold". The two blobs represent possible points and the associated error, describing the protocol: the purple sphere is centred on the theoretical point representing the correlation stemming from the actual resource used (Bell pair); in green there is the point that would be obtained in the absence of any correlations. It's worth mentioning that the nearest point on the surface to the purple point, is not the green one.}
    \label{fig:blob}
\end{figure}

\subsection{Numerical example}

To explain more in detail how the test works, let's consider the information available in theory to Ashley and Charlie: $a_0, a_1, c_0, c1$ (using the notation introduced previously). Consider two clearly distinguishable situations: the correlations measured if a Bell state is shared, and the correlation obtained if two uncorrelated states are shared. In the first case the multiple correlation are $a_0c_0=a_1c_1=1/2$, while in the second (uncorrelated) case, the pairwise correlations are just the products of the probabilities of the single events $a_0c_0=a_1c_0=a_0c_1=a_1c_1=1/4$. The idea of the test is to consider the distance in the 8d dimension space between those 2 points with the experimental errors and perform a separability test. The two points will be called $P_u$ (uncorrelated) and $P_b$ (Bell pair). 
To continue with the example let's find a reasonable error to associate to the expected point $P_b$. A new simulated noisy point  $\widetilde{P}_b$, is found by adding a Gaussian error with standard deviation of 5\% on each component. The distance between the ideal point and the one obtained altering each component by exactly 5\%, is taken as the standard deviation of the Gaussian distribution describing the error on the distance between $P_u$ and $\widetilde{P}_b$. In this way standard 2 Gaussian separability test can be performed. The same can be made only adding noise on $P_u$, obtaining $\widetilde{P}_u$. The two situation are described in figure \ref{fig:gaussian_test}.

\begin{figure}[H]
    \centering
    \begin{subfigure}[b]{0.45\textwidth}
    \includegraphics[width=\textwidth]{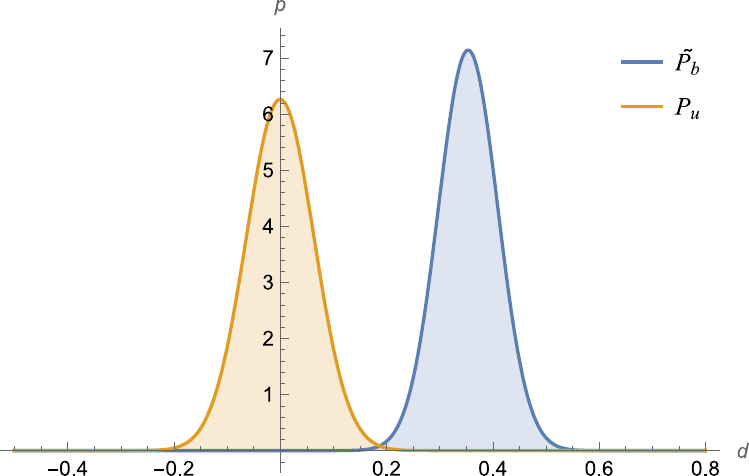}
    \end{subfigure}
    \hfill
    \begin{subfigure}[b]{0.45\textwidth}
    \includegraphics[width=\textwidth]{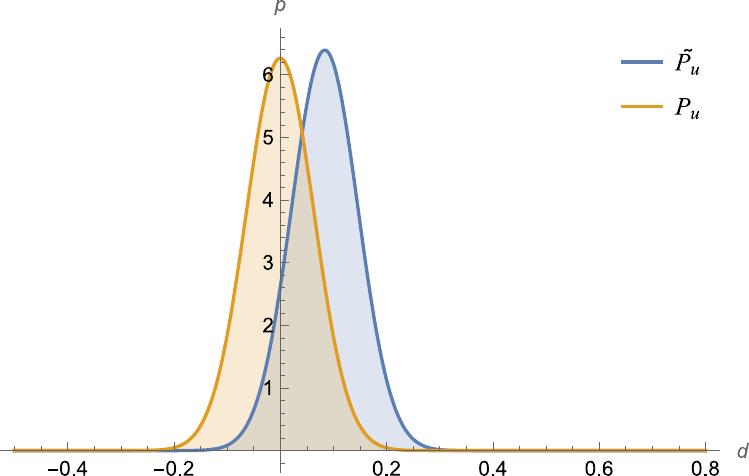}
    \end{subfigure}
    \caption{In orange, centred on 0, is the point $P_u$ and the distribution representing its error along the line connecting with the second point. On the left figure, in blue, the second point is the noisy $\widetilde{P_b}$, describing the case obtained with a shared Bell pair; on the right, in blue, the second point is the noisy $\widetilde{P_u}$, obtaining by adding noise to the first point. Computing the overlapping, the area common to the two Gaussians, or performing a separability test, confirm that in the first case the 2 points are indeed different, while in the second case (depending on the significance level) the hypothesis test do not allows to separate them.}
    \label{fig:gaussian_test}
\end{figure}

In an experiment, the same steps can be adapted by considering the available information and the experimental precision and measured quantities. In addition, it could be preferable to execute the same test on samples instead than recovering the parameters of the error distributions. In this case, standard hypothesis testing techniques such as a Two-Sample T-Test or a Two-sample Kolmogorov-Smirnov test \cite{pratt1981kolmogorov, cressie1986use} can be used.

\subsection{Certifying nonclassicality of behaviours}

We would like to better understand the relation between the geometrical distance in the behaviour space and some known metrics used in quantum information theory.
Consider the expected and observed behaviour, defined as points $P$ and $Q$ in behaviour space as in equation \eqref{eqn:8dvertex_structure}. 
Consider the difference,
\begin{equation}
    V=P-Q.
\end{equation}
The norm of this vector is what is used in our statistical test to understand if the two behaviours are coincident or not. As mentioned, in principle it could used any norm. A standard example is the class of the $L_p$ norms defined as,

\begin{equation}
    \|x\|_p = \left(|x_1|^p + |x_2|^p + \dotsb + |x_n|^p\right)^{1/p},
\end{equation}
\noindent where $x\in \mathbb{R}^n$ and $p\geq 1$.
For future reference, it is useful to note that, for vector with entries $x_i\in [0,1]$, the following is true,

\begin{equation} \label{eq:L2L1relation}
    \sum_i |x_i|^2 \leq \sum_i |x_i| \implies  ||x||_2\leq ||x||_1.
\end{equation}

To precisely characterise the observed behaviour, $N$ runs are repeated and the observed point is obtained as the average of behaviours, with their standard deviation. 
To start the investigation, consider that the elements of the vector characterising the observed behaviour at run $k$ are conditional probabilities of the form,

\begin{equation}
    P(ab|xy)=\Tr \left[ M^x \otimes M^y \sigma_k \right] = \Tr \left[  \Pi \sigma_k \right],
\end{equation}

\noindent where $a,b$ are the measurement outcomes, $x,y$ the choice of measurement basis and the measurement operator, for clarity of notation, was indicated as a projector $\Pi$.  
We also know that,

\begin{equation}\label{eq:NeilsChuang9.27}
    \left| \Tr (\mathcal{O} \rho) \right| \leq \Tr (\mathcal{O} |\rho|),
\end{equation}
\noindent for a generic density matrix $\rho$ and a bounded positive semi-definite Hermitian operator $\mathcal{O}$, as shown in \cite{Nielsen_Chuang_2010}. This allows us to consider a single element of $V$ and bound it. In particular, denoting the state of the system at run $k$ by $\sigma_k$, and by $\rho=\ket{\psi}\bra{\psi}$, the ideal and pure expected state,

\begin{equation}
    \left|\frac{1}{N}\sum_k \Tr (\Pi \sigma_k) - \Tr (\Pi \rho) \right| \leq \Tr (\Pi |\sigma - \rho |),
\end{equation}

\noindent where $\sigma$ is the average density matrix, the inequality follows from equation \eqref{eq:NeilsChuang9.27}, and we follow the standard i.i.d. assumptions. 

Using the previous observation we want to bound the geometrical distance between behaviours, i.e. the $L_2$ norm of $V$. In \cite{Nielsen_Chuang_2010} it was shown that, given two states, for any POVM  $\{E_i\}_i$ the classical trace distance $D$ (also called the total variation distance) between the probability distributions obtained by measuring $p_m=\Tr (E_m \rho)$, 
$q_m=\Tr (E_m \sigma)$, is bounded by the trace distance $\delta$ of between the density matrices,

\begin{align}
    D(p_m,q_m)&=\frac{1}{2}\sum_m \left| p_m - q_m\right| = \frac{1}{2}\sum_m |\Tr \left( E_m (\rho - \sigma)\right)\\ &\leq \frac{1}{2}\sum_m \Tr 
    \left[ E_m \left| \rho - \sigma \right| \right] =\delta (\rho,\sigma).
\end{align}

We then introduce $p_i=\Tr(\Pi_i \rho)$ and $q_i=\Tr(\Pi_i \sigma)$, the classical distributions obtained by measuring the states. 
We also consider the reduced density matrices obtained by tracing out the two parties: $\sigma_A, \sigma_B, \rho_A, \rho_B$. Consider the norm $||V||_1$, it can then be written as a sum containing the measurement operators. Since each term is positive, we can bound the sum adding the remaining operators to complete three set of POVM, $I=\{\Pi_i\}_i$, $J=\{\Pi_j\}_j$, $L=\{\Pi_l\}_l$, for measuring the complete system and each subsystem,
\begin{align}
    || V ||_1 
    &= \sum_i \left| \Tr (\Pi_i^A (\rho_A - \sigma_A)) \right| + \sum_j \left| \Tr (\Pi_j^B (\rho_B - \sigma_B)) \right| + \sum_l \left| \Tr (\Pi_l^{AB} (\rho_{AB} - \sigma_{AB})) \right| \\
    &\leq \sum_{i \in I} \left| \Tr (\Pi_i^A (\rho_A - \sigma_A)) \right| + \sum_{j \in J} \left| \Tr (\Pi_j^B (\rho_B - \sigma_B)) \right| + \sum_{l \in L} \left| \Tr (\Pi_l^{AB} (\rho_{AB} - \sigma_{AB})) \right| \\
    &\leq \sum_{i \in I} \Tr \left[\Pi_i^A \left|\rho_A - \sigma_A\right| \right] + \sum_{j \in J} \Tr \left[\Pi_j^B \left|\rho_B - \sigma_B\right| \right] + \sum_{l \in L} \Tr \left[\Pi_l^A \left|\rho_{AB} - \sigma_{AB}\right| \right] \\
    &\leq 2 \left[ \delta(\rho_A, \sigma_A) + \delta(\rho_B, \sigma_B) + \delta(\rho_{AB}, \sigma_{AB}) \right] . \label{eq:L1<traces}
\end{align}

Consider now Euclidean distance between the behaviours, $||V||_2$.
Each element, being a probability, is between $0$ and $1$. Therefore, combining equations \eqref{eq:L2L1relation} and \eqref{eq:L1<traces} it holds,

\begin{equation}
    || V ||_2 \leq || V ||_1 \leq 2 \left[ \delta(\rho_A, \sigma_A) + \delta(\rho_B, \sigma_B) + \delta(\rho_{AB}, \sigma_{AB}) \right].
\end{equation}

\noindent  This inequality explain the relation between the geometrical distance in the behaviour space and the trace distance between the expected and observed states. It is worth noting that the trace distance is related to the probability of distinguishing between two quantum states, therefore, intuitively, it makes sense that it appears in the statistical test. Also, similar techniques were employed in \cite{PhysRevA.80.062327, PhysRevA.97.022111}.

This result support the previous geometric analysis as the combination of trace distances is bound from below by the euclidean distance between behaviours (and more strictly by $||V||_1$). The connection to fidelity $F$ can be recovered by the well known bounds,

\begin{equation}
    1-\sqrt{F(\rho,\sigma)} \le T(\rho,\sigma) \le\sqrt{1-F(\rho,\sigma)}.
\end{equation}

This proof, as well as the previous discussion and the statistical test, do not depend on the actual states used during the protocol. Therefore the techniques presented here can be useful in more general scenarios than BB84, employing non standard states (such as Bell or GHZ states commonly used) and sharing only the causal structure. Since this computation do not use the causal structure, similar statistical tests can be devised for the causal structure under consideration. 
In addition, this structure works in purely classical scenarios, as those can be recovered from density matrix formalism.

\section{Discussion and conclusion}

In contrast to our proposal to use a statistical test, another line of reasoning to obtain a unit of nonclassicality can be taken. Suppose all the relevant data on the protocol are available at the end of the run, in particular all the setting variables and measurement outcomes relative to both parties. 
Using the data on the measurements performed, a quantum tomography \cite{christandl2012reliable} can be performed. The reconstructed channel can then be compared to the expected one using fidelity other measures, such as the Bures distance or Bures arc \cite{bures1969extension}. 
On the other hand our approach is inspired by statistical techniques for quantifying nonclassicality \cite{reviewStatisticalMethodsQuantumVerification}, and the need for a sampling in both methods highlight a connection.
Establishing a precise correlation between statistical quantities derived from the test we propose and quantum information quantities derived from the applying quantum tomography depends on the protocol and falls beyond the scope of this discussion.


Our contributions aim to find a contact point between resource theory, describing Network scenarios, and QKD. From this union it is possible to find a new description of QKD protocols, where the effects of noise or the presence of an eavesdropper are encoded in a geometrical manner. Along the way we also analyse a non convex shape that naturally arise in this context. Intriguingly, this approach enables the estimation of process fidelity for the reconstructed channel, alongside other fidelity-derived metrics such as the Bures distance, which similarly serve as measures of nonclassicality.
Note that these comparative analyses illuminate both the differences and similarities between the two methodologies, thereby connecting our proposal to established techniques in Quantum Key Distribution (QKD).

Since the techniques discussed are not directly dependent on the specific quantum states employed in the protocol, we emphasise the broad applicability of them. Consequently, these methods transcend the BB84 protocol and are applicable to scenarios involving non-standard states, such as those beyond Bell or GHZ states, provided they adhere to the same causal structure. Moreover, since the computation does not depend on the causal structure, similar statistical tests can be devised for other causal structures under investigation.

This study employed the terminology of resource theory to analyse scenarios similar to those encountered in typical QKD tasks. During the analysis, certain geometrical problems emerged, suggesting that future research could provide a more thorough characterisation of the non-convex polytope under investigation. On the QKD side, the proposed test is a theoretical experiment designed to link a practical problem with its abstract representation in behaviour space.

In conclusion, this work prompt further studies also on the geometrical part, both on the full and the reduced polytope. The goal on this new investigation could be trying to find the facet equations, as well as a general description, of the non convex geometrical shape. It can also be interesting to understand, from an information-theoretical standpoint,  the significance of behaviours in the region between the standard convex hull of the extremal strategies, and the non convex geometry resulting from the application of the causal structures.

\paragraph{Acknowledgments.}  We thank Rafael Wagner and Yoann Piétri for helpful discussions. MdO is supported by National Funds through the FCT - Fundação para a Ciência e a Tecnologia, I.P. (Portuguese Foundation for Science and Technology) within the project IBEX, with reference PTDC/CCI-COM/4280/2021, and via CEECINST/00062/2018 (EFG). Andrea D’Urbano acknowledges the funding received by Deep Consulting
s.r.l. within the Ph.D. program.


\bibliographystyle{unsrt}
\bibliography{refS}

\clearpage
\appendix

\section{Appendix: Vertices list} \label{appendix:A}

In table \ref{tab:vertices} are reported the vertices of the polytope describing the correlation scenario P3. The entries are divided to highlight the free variables describing single events, and the composite events of pairs and triplets. The convention used is described from equation \ref{eqn:det_strategies} to equation \ref{eqn:vertex_structure}. In table \ref{tab:8d_vertices} are listed the vertices obtained through the marginalisation of the knowledge of the middle party, Blake. The structure of the vertices is reported is equation \ref{eqn:8dvertex_structure}.

To have a better idea of the structure explored it can be constructed a graph by using the vertices as nodes; two nodes are then connected by an edge if and only if the corresponding vertices can see each others. 
To have a visual representation of such graphs, a projection in 3 dimension can be used, and in order to maximise the distinguishability of features the first three components of the singular value decomposition were used. As a side note, observe that for numerical stability, it was added a shift to the zero vector representing the first vertex.
The corresponding graphs are depicted in figure \ref{fig:graph_vertices} for both polytopes.

\begin{table}
    \centering
    \begin{adjustbox}{width=0.85\textwidth}   
    {\small
    \begin{tabular}{|cccccc|cccccccccccc|cccccccc|} 
        \hline
        \multicolumn{26}{|c|}{Vertices} \\
        \hline
        \multicolumn{6}{|c|}{Single events} & \multicolumn{12}{c|}{Pairs events} & \multicolumn{8}{c|}{Triplets events} \\
        \hline
0 & 0 & 0 & 0 & 0 & 0 & 0 & 0 & 0 & 0 & 0 & 0 & 0 & 0 & 0 & 0 & 0 & 0 & 0 & 0 & 0 & 0 & 0 & 0 & 0 & 0 \\ \hline
0 & 0 & 0 & 0 & 0 & 1 & 0 & 0 & 0 & 0 & 0 & 0 & 0 & 0 & 0 & 0 & 0 & 0 & 0 & 0 & 0 & 0 & 0 & 0 & 0 & 0 \\ \hline
0 & 0 & 0 & 0 & 1 & 0 & 0 & 0 & 0 & 0 & 0 & 0 & 0 & 0 & 0 & 0 & 0 & 0 & 0 & 0 & 0 & 0 & 0 & 0 & 0 & 0 \\ \hline
0 & 0 & 0 & 0 & 1 & 1 & 0 & 0 & 0 & 0 & 0 & 0 & 0 & 0 & 0 & 0 & 0 & 0 & 0 & 0 & 0 & 0 & 0 & 0 & 0 & 0 \\ \hline
0 & 0 & 0 & 1 & 0 & 0 & 0 & 0 & 0 & 0 & 0 & 0 & 0 & 0 & 0 & 0 & 0 & 0 & 0 & 0 & 0 & 0 & 0 & 0 & 0 & 0 \\ \hline
0 & 0 & 0 & 1 & 0 & 1 & 0 & 0 & 0 & 0 & 0 & 0 & 0 & 0 & 0 & 0 & 0 & 1 & 0 & 0 & 0 & 0 & 0 & 0 & 0 & 0 \\ \hline
0 & 0 & 0 & 1 & 1 & 0 & 0 & 0 & 0 & 0 & 0 & 0 & 0 & 0 & 0 & 0 & 1 & 0 & 0 & 0 & 0 & 0 & 0 & 0 & 0 & 0 \\ \hline
0 & 0 & 0 & 1 & 1 & 1 & 0 & 0 & 0 & 0 & 0 & 0 & 0 & 0 & 0 & 0 & 1 & 1 & 0 & 0 & 0 & 0 & 0 & 0 & 0 & 0 \\ \hline
0 & 0 & 1 & 0 & 0 & 0 & 0 & 0 & 0 & 0 & 0 & 0 & 0 & 0 & 0 & 0 & 0 & 0 & 0 & 0 & 0 & 0 & 0 & 0 & 0 & 0 \\ \hline
0 & 0 & 1 & 0 & 0 & 1 & 0 & 0 & 0 & 0 & 0 & 0 & 0 & 0 & 0 & 1 & 0 & 0 & 0 & 0 & 0 & 0 & 0 & 0 & 0 & 0 \\ \hline
0 & 0 & 1 & 0 & 1 & 0 & 0 & 0 & 0 & 0 & 0 & 0 & 0 & 0 & 1 & 0 & 0 & 0 & 0 & 0 & 0 & 0 & 0 & 0 & 0 & 0 \\ \hline
0 & 0 & 1 & 0 & 1 & 1 & 0 & 0 & 0 & 0 & 0 & 0 & 0 & 0 & 1 & 1 & 0 & 0 & 0 & 0 & 0 & 0 & 0 & 0 & 0 & 0 \\ \hline
0 & 0 & 1 & 1 & 0 & 0 & 0 & 0 & 0 & 0 & 0 & 0 & 0 & 0 & 0 & 0 & 0 & 0 & 0 & 0 & 0 & 0 & 0 & 0 & 0 & 0 \\ \hline
0 & 0 & 1 & 1 & 0 & 1 & 0 & 0 & 0 & 0 & 0 & 0 & 0 & 0 & 0 & 1 & 0 & 1 & 0 & 0 & 0 & 0 & 0 & 0 & 0 & 0 \\ \hline
0 & 0 & 1 & 1 & 1 & 0 & 0 & 0 & 0 & 0 & 0 & 0 & 0 & 0 & 1 & 0 & 1 & 0 & 0 & 0 & 0 & 0 & 0 & 0 & 0 & 0 \\ \hline
0 & 0 & 1 & 1 & 1 & 1 & 0 & 0 & 0 & 0 & 0 & 0 & 0 & 0 & 1 & 1 & 1 & 1 & 0 & 0 & 0 & 0 & 0 & 0 & 0 & 0 \\ \hline
0 & 1 & 0 & 0 & 0 & 0 & 0 & 0 & 0 & 0 & 0 & 0 & 0 & 0 & 0 & 0 & 0 & 0 & 0 & 0 & 0 & 0 & 0 & 0 & 0 & 0 \\ \hline
0 & 1 & 0 & 0 & 0 & 1 & 0 & 0 & 0 & 0 & 0 & 0 & 0 & 1 & 0 & 0 & 0 & 0 & 0 & 0 & 0 & 0 & 0 & 0 & 0 & 0 \\ \hline
0 & 1 & 0 & 0 & 1 & 0 & 0 & 0 & 0 & 0 & 0 & 0 & 1 & 0 & 0 & 0 & 0 & 0 & 0 & 0 & 0 & 0 & 0 & 0 & 0 & 0 \\ \hline
0 & 1 & 0 & 0 & 1 & 1 & 0 & 0 & 0 & 0 & 0 & 0 & 1 & 1 & 0 & 0 & 0 & 0 & 0 & 0 & 0 & 0 & 0 & 0 & 0 & 0 \\ \hline
0 & 1 & 0 & 1 & 0 & 0 & 0 & 0 & 0 & 1 & 0 & 0 & 0 & 0 & 0 & 0 & 0 & 0 & 0 & 0 & 0 & 0 & 0 & 0 & 0 & 0 \\ \hline
0 & 1 & 0 & 1 & 0 & 1 & 0 & 0 & 0 & 1 & 0 & 0 & 0 & 1 & 0 & 0 & 0 & 1 & 0 & 0 & 0 & 0 & 0 & 0 & 0 & 1 \\ \hline
0 & 1 & 0 & 1 & 1 & 0 & 0 & 0 & 0 & 1 & 0 & 0 & 1 & 0 & 0 & 0 & 1 & 0 & 0 & 0 & 0 & 0 & 0 & 0 & 1 & 0 \\ \hline
0 & 1 & 0 & 1 & 1 & 1 & 0 & 0 & 0 & 1 & 0 & 0 & 1 & 1 & 0 & 0 & 1 & 1 & 0 & 0 & 0 & 0 & 0 & 0 & 1 & 1 \\ \hline
0 & 1 & 1 & 0 & 0 & 0 & 0 & 0 & 1 & 0 & 0 & 0 & 0 & 0 & 0 & 0 & 0 & 0 & 0 & 0 & 0 & 0 & 0 & 0 & 0 & 0 \\ \hline
0 & 1 & 1 & 0 & 0 & 1 & 0 & 0 & 1 & 0 & 0 & 0 & 0 & 1 & 0 & 1 & 0 & 0 & 0 & 0 & 0 & 0 & 0 & 1 & 0 & 0 \\ \hline
0 & 1 & 1 & 0 & 1 & 0 & 0 & 0 & 1 & 0 & 0 & 0 & 1 & 0 & 1 & 0 & 0 & 0 & 0 & 0 & 0 & 0 & 1 & 0 & 0 & 0 \\ \hline
0 & 1 & 1 & 0 & 1 & 1 & 0 & 0 & 1 & 0 & 0 & 0 & 1 & 1 & 1 & 1 & 0 & 0 & 0 & 0 & 0 & 0 & 1 & 1 & 0 & 0 \\ \hline
0 & 1 & 1 & 1 & 0 & 0 & 0 & 0 & 1 & 1 & 0 & 0 & 0 & 0 & 0 & 0 & 0 & 0 & 0 & 0 & 0 & 0 & 0 & 0 & 0 & 0 \\ \hline
0 & 1 & 1 & 1 & 0 & 1 & 0 & 0 & 1 & 1 & 0 & 0 & 0 & 1 & 0 & 1 & 0 & 1 & 0 & 0 & 0 & 0 & 0 & 1 & 0 & 1 \\ \hline
0 & 1 & 1 & 1 & 1 & 0 & 0 & 0 & 1 & 1 & 0 & 0 & 1 & 0 & 1 & 0 & 1 & 0 & 0 & 0 & 0 & 0 & 1 & 0 & 1 & 0 \\ \hline
0 & 1 & 1 & 1 & 1 & 1 & 0 & 0 & 1 & 1 & 0 & 0 & 1 & 1 & 1 & 1 & 1 & 1 & 0 & 0 & 0 & 0 & 1 & 1 & 1 & 1 \\ \hline
1 & 0 & 0 & 0 & 0 & 0 & 0 & 0 & 0 & 0 & 0 & 0 & 0 & 0 & 0 & 0 & 0 & 0 & 0 & 0 & 0 & 0 & 0 & 0 & 0 & 0 \\ \hline
1 & 0 & 0 & 0 & 0 & 1 & 0 & 0 & 0 & 0 & 0 & 1 & 0 & 0 & 0 & 0 & 0 & 0 & 0 & 0 & 0 & 0 & 0 & 0 & 0 & 0 \\ \hline
1 & 0 & 0 & 0 & 1 & 0 & 0 & 0 & 0 & 0 & 1 & 0 & 0 & 0 & 0 & 0 & 0 & 0 & 0 & 0 & 0 & 0 & 0 & 0 & 0 & 0 \\ \hline
1 & 0 & 0 & 0 & 1 & 1 & 0 & 0 & 0 & 0 & 1 & 1 & 0 & 0 & 0 & 0 & 0 & 0 & 0 & 0 & 0 & 0 & 0 & 0 & 0 & 0 \\ \hline
1 & 0 & 0 & 1 & 0 & 0 & 0 & 1 & 0 & 0 & 0 & 0 & 0 & 0 & 0 & 0 & 0 & 0 & 0 & 0 & 0 & 0 & 0 & 0 & 0 & 0 \\ \hline
1 & 0 & 0 & 1 & 0 & 1 & 0 & 1 & 0 & 0 & 0 & 1 & 0 & 0 & 0 & 0 & 0 & 1 & 0 & 0 & 0 & 1 & 0 & 0 & 0 & 0 \\ \hline
1 & 0 & 0 & 1 & 1 & 0 & 0 & 1 & 0 & 0 & 1 & 0 & 0 & 0 & 0 & 0 & 1 & 0 & 0 & 0 & 1 & 0 & 0 & 0 & 0 & 0 \\ \hline
1 & 0 & 0 & 1 & 1 & 1 & 0 & 1 & 0 & 0 & 1 & 1 & 0 & 0 & 0 & 0 & 1 & 1 & 0 & 0 & 1 & 1 & 0 & 0 & 0 & 0 \\ \hline
1 & 0 & 1 & 0 & 0 & 0 & 1 & 0 & 0 & 0 & 0 & 0 & 0 & 0 & 0 & 0 & 0 & 0 & 0 & 0 & 0 & 0 & 0 & 0 & 0 & 0 \\ \hline
1 & 0 & 1 & 0 & 0 & 1 & 1 & 0 & 0 & 0 & 0 & 1 & 0 & 0 & 0 & 1 & 0 & 0 & 0 & 1 & 0 & 0 & 0 & 0 & 0 & 0 \\ \hline
1 & 0 & 1 & 0 & 1 & 0 & 1 & 0 & 0 & 0 & 1 & 0 & 0 & 0 & 1 & 0 & 0 & 0 & 1 & 0 & 0 & 0 & 0 & 0 & 0 & 0 \\ \hline
1 & 0 & 1 & 0 & 1 & 1 & 1 & 0 & 0 & 0 & 1 & 1 & 0 & 0 & 1 & 1 & 0 & 0 & 1 & 1 & 0 & 0 & 0 & 0 & 0 & 0 \\ \hline
1 & 0 & 1 & 1 & 0 & 0 & 1 & 1 & 0 & 0 & 0 & 0 & 0 & 0 & 0 & 0 & 0 & 0 & 0 & 0 & 0 & 0 & 0 & 0 & 0 & 0 \\ \hline
1 & 0 & 1 & 1 & 0 & 1 & 1 & 1 & 0 & 0 & 0 & 1 & 0 & 0 & 0 & 1 & 0 & 1 & 0 & 1 & 0 & 1 & 0 & 0 & 0 & 0 \\ \hline
1 & 0 & 1 & 1 & 1 & 0 & 1 & 1 & 0 & 0 & 1 & 0 & 0 & 0 & 1 & 0 & 1 & 0 & 1 & 0 & 1 & 0 & 0 & 0 & 0 & 0 \\ \hline
1 & 0 & 1 & 1 & 1 & 1 & 1 & 1 & 0 & 0 & 1 & 1 & 0 & 0 & 1 & 1 & 1 & 1 & 1 & 1 & 1 & 1 & 0 & 0 & 0 & 0 \\ \hline
1 & 1 & 0 & 0 & 0 & 0 & 0 & 0 & 0 & 0 & 0 & 0 & 0 & 0 & 0 & 0 & 0 & 0 & 0 & 0 & 0 & 0 & 0 & 0 & 0 & 0 \\ \hline
1 & 1 & 0 & 0 & 0 & 1 & 0 & 0 & 0 & 0 & 0 & 1 & 0 & 1 & 0 & 0 & 0 & 0 & 0 & 0 & 0 & 0 & 0 & 0 & 0 & 0 \\ \hline
1 & 1 & 0 & 0 & 1 & 0 & 0 & 0 & 0 & 0 & 1 & 0 & 1 & 0 & 0 & 0 & 0 & 0 & 0 & 0 & 0 & 0 & 0 & 0 & 0 & 0 \\ \hline
1 & 1 & 0 & 0 & 1 & 1 & 0 & 0 & 0 & 0 & 1 & 1 & 1 & 1 & 0 & 0 & 0 & 0 & 0 & 0 & 0 & 0 & 0 & 0 & 0 & 0 \\ \hline
1 & 1 & 0 & 1 & 0 & 0 & 0 & 1 & 0 & 1 & 0 & 0 & 0 & 0 & 0 & 0 & 0 & 0 & 0 & 0 & 0 & 0 & 0 & 0 & 0 & 0 \\ \hline
1 & 1 & 0 & 1 & 0 & 1 & 0 & 1 & 0 & 1 & 0 & 1 & 0 & 1 & 0 & 0 & 0 & 1 & 0 & 0 & 0 & 1 & 0 & 0 & 0 & 1 \\ \hline
1 & 1 & 0 & 1 & 1 & 0 & 0 & 1 & 0 & 1 & 1 & 0 & 1 & 0 & 0 & 0 & 1 & 0 & 0 & 0 & 1 & 0 & 0 & 0 & 1 & 0 \\ \hline
1 & 1 & 0 & 1 & 1 & 1 & 0 & 1 & 0 & 1 & 1 & 1 & 1 & 1 & 0 & 0 & 1 & 1 & 0 & 0 & 1 & 1 & 0 & 0 & 1 & 1 \\ \hline
1 & 1 & 1 & 0 & 0 & 0 & 1 & 0 & 1 & 0 & 0 & 0 & 0 & 0 & 0 & 0 & 0 & 0 & 0 & 0 & 0 & 0 & 0 & 0 & 0 & 0 \\ \hline
1 & 1 & 1 & 0 & 0 & 1 & 1 & 0 & 1 & 0 & 0 & 1 & 0 & 1 & 0 & 1 & 0 & 0 & 0 & 1 & 0 & 0 & 0 & 1 & 0 & 0 \\ \hline
1 & 1 & 1 & 0 & 1 & 0 & 1 & 0 & 1 & 0 & 1 & 0 & 1 & 0 & 1 & 0 & 0 & 0 & 1 & 0 & 0 & 0 & 1 & 0 & 0 & 0 \\ \hline
1 & 1 & 1 & 0 & 1 & 1 & 1 & 0 & 1 & 0 & 1 & 1 & 1 & 1 & 1 & 1 & 0 & 0 & 1 & 1 & 0 & 0 & 1 & 1 & 0 & 0 \\ \hline
1 & 1 & 1 & 1 & 0 & 0 & 1 & 1 & 1 & 1 & 0 & 0 & 0 & 0 & 0 & 0 & 0 & 0 & 0 & 0 & 0 & 0 & 0 & 0 & 0 & 0 \\ \hline
1 & 1 & 1 & 1 & 0 & 1 & 1 & 1 & 1 & 1 & 0 & 1 & 0 & 1 & 0 & 1 & 0 & 1 & 0 & 1 & 0 & 1 & 0 & 1 & 0 & 1 \\ \hline
1 & 1 & 1 & 1 & 1 & 0 & 1 & 1 & 1 & 1 & 1 & 0 & 1 & 0 & 1 & 0 & 1 & 0 & 1 & 0 & 1 & 0 & 1 & 0 & 1 & 0 \\ \hline
1 & 1 & 1 & 1 & 1 & 1 & 1 & 1 & 1 & 1 & 1 & 1 & 1 & 1 & 1 & 1 & 1 & 1 & 1 & 1 & 1 & 1 & 1 & 1 & 1 & 1 \\ \hline
    \end{tabular}
    }
    \end{adjustbox}
    \caption{List of vertices of the polytope describing the correlation scenario P3.}
    \label{tab:vertices}
\end{table}

\begin{table}
    \centering
    {\small
    \begin{tabular}{|cccc|cccc|} 
        \hline
        \multicolumn{8}{|c|}{Vertices} \\
        \hline
        \multicolumn{4}{|c|}{Single events} & \multicolumn{4}{c|}{Pairs events} \\
        \hline
 0 & 0 & 0 & 0 & 0 & 0 & 0 & 0 \\
 0 & 0 & 0 & 1 & 0 & 0 & 0 & 0 \\
 0 & 0 & 1 & 0 & 0 & 0 & 0 & 0 \\
 0 & 0 & 1 & 1 & 0 & 0 & 0 & 0 \\
 0 & 1 & 0 & 0 & 0 & 0 & 0 & 0 \\
 0 & 1 & 0 & 1 & 0 & 0 & 0 & 1 \\
 0 & 1 & 1 & 0 & 0 & 0 & 1 & 0 \\
 0 & 1 & 1 & 1 & 0 & 0 & 1 & 1 \\
 1 & 0 & 0 & 0 & 0 & 0 & 0 & 0 \\
 1 & 0 & 0 & 1 & 0 & 1 & 0 & 0 \\
 1 & 0 & 1 & 0 & 1 & 0 & 0 & 0 \\
 1 & 0 & 1 & 1 & 1 & 1 & 0 & 0 \\
 1 & 1 & 0 & 0 & 0 & 0 & 0 & 0 \\
 1 & 1 & 0 & 1 & 0 & 1 & 0 & 1 \\
 1 & 1 & 1 & 0 & 1 & 0 & 1 & 0 \\
 1 & 1 & 1 & 1 & 1 & 1 & 1 & 1 \\ \hline
    \end{tabular}
    }
    \caption{List of vertices of the polytope describing the correlation scenario P3, with the marginalisation of the knowledge of the middle party.}
    \label{tab:8d_vertices}
\end{table}

\begin{figure}[H]
    \centering
    \includegraphics[width=0.45\textwidth]{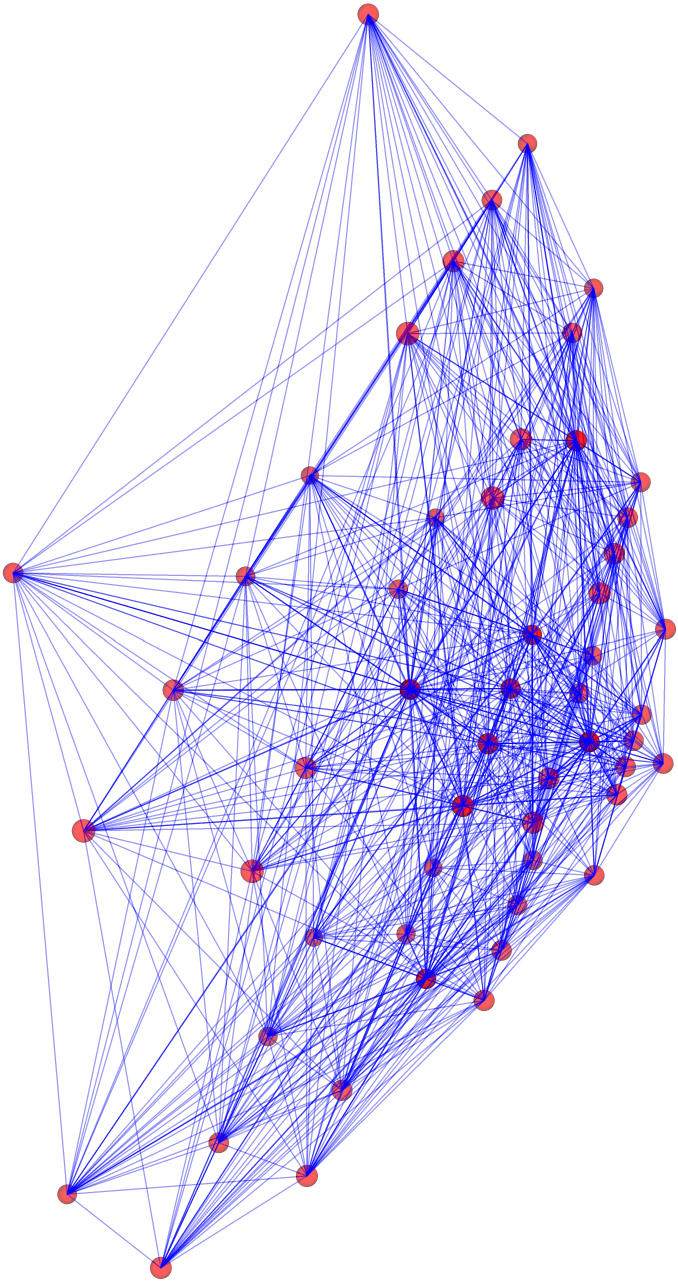}
    \includegraphics[width=0.45\textwidth]{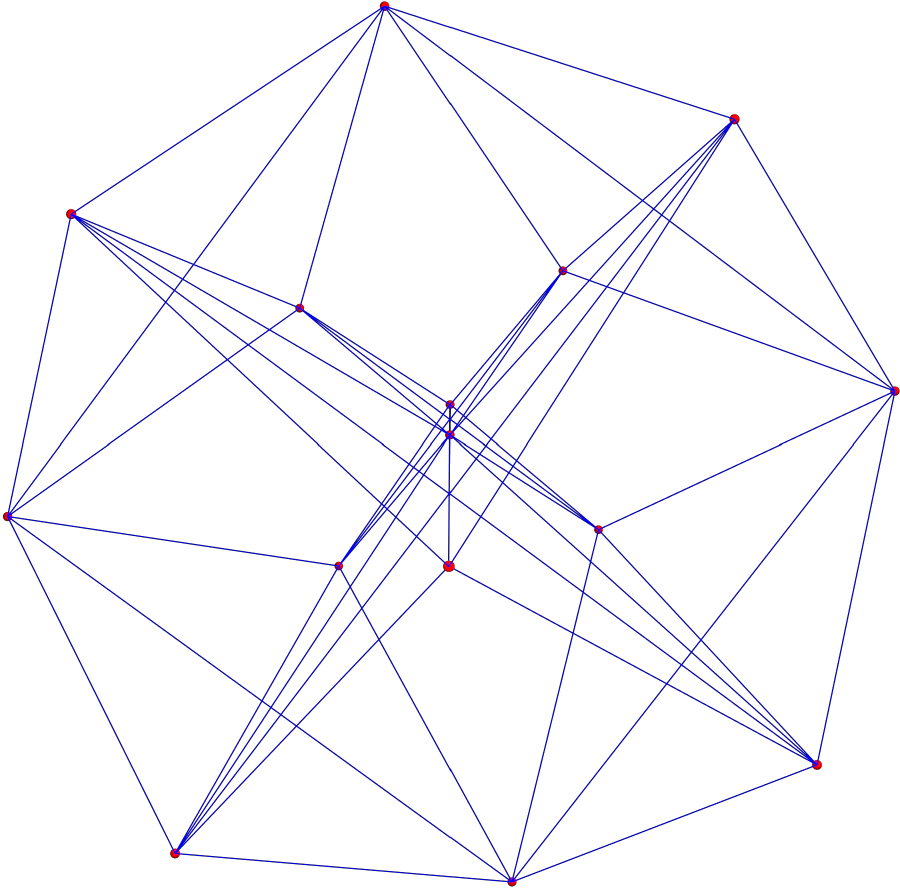}
    \caption{A graph is constructed using the vertices of the 26D polytope (on left), or in 8D (on right) as nodes, with two nodes connected by an edge if and only if the corresponding vertices have visibility of each other. The 3 dimensional coordinates are obtained using singular value decomposition.}
    \label{fig:graph_vertices}
\end{figure}

\section{Appendix: Covering problems} \label{appendix:B}

In this appendix is proved constructively that the minimum number of generators, i.e. vertices generating edges through theorem \ref{theorem:test}, is equal to the lower bound of 4.
As in theorem \ref{theorem:vertices_generators}, the deterministic strategies in equation \ref{eqn:det_strategies} will be denoted as a triplet $(\alpha_i,\beta_j,\gamma_k)$ where $i,j,k \in {0,1,2,3}$.
As proved in the theorem, the matrix $(\alpha_i,\gamma_k)$ can be used to describe the effect of the test described in theorem \ref{theorem:test}. A vertex belonging to the class described by an element of this matrix, can see all the other class belonging to the same column or row: for example $(\alpha_1, \beta_0,\gamma_3)$ belongs to the class $(\alpha_1,\gamma_3)$ and can see all the vertices in classes $(\alpha_i,\gamma_3)$ and $(\alpha_1,\gamma_j)$.
Using this description, in table \ref{tab:shapes1} is found an explicit construction of the generators from diagonal elements, while a second construction, is found in table \ref{tab:shapes2}. 
The same steps are followed for the marginalisation of the middle party resulting in an 8 dimensional reduced non convex polytope, reported in tables \ref{tab:8d_shapes1} and \ref{tab:8d_shapes2}. In this case each class is composed by exactly one vertex: the marginalisation process amounts precisely to collapsing the information related to $\beta_i$.

\begin{table}
    \centering
    \begin{adjustbox}{width=\textwidth}   
    \begin{tabular}{||ccc|c||ccc|c||ccc|c||ccc|c||c||} 
        \hline
        \multicolumn{4}{||c||}{Set 1} & \multicolumn{4}{c||}{Set 2} & \multicolumn{4}{c||}{Set 3} & \multicolumn{4}{c||}{Set 4} &  \\
        \hline
        \multicolumn{3}{||c|}{vertices} & \# & \multicolumn{3}{c|}{vertices} & \# & \multicolumn{3}{c|}{vertices} & \# & \multicolumn{3}{c|}{vertices} & \# & Tot \\
        \hline
        $\alpha_0$ & $\beta_0$ & $\gamma_0$ & 1 & $\alpha_1$ & $\beta_1$ & $\gamma_1$ & 1 & $\alpha_2$ & $\beta_2$ & $\gamma_2$ & 1 & $\alpha_3$ & $\beta_3$ & $\gamma_3$ & 1 & \\ 
        $\alpha_0$ & $\beta_0$ & x & 3 & $\alpha_1$ & $\beta_1$ & x & 2 & $\alpha_2$ & $\beta_2$ & x & 1 & $\alpha_3$ & $\beta_3$ & x & 0 & \\
        $\alpha_0$ & x & $\gamma_0$ & 3 & $\alpha_1$ & x & $\gamma_1$ & 3 & $\alpha_2$ & x & $\gamma_2$ & 3 & $\alpha_3$ & x & $\gamma_3$ & 3 & \\
        x & $\beta_0$ & $\gamma_0$ & 3 & x & $\beta_1$ & $\gamma_1$ & 2 & x & $\beta_2$ & $\gamma_2$ & 1 & x & $\beta_3$ & $\gamma_3$ & 0 & \\
        $\alpha_0$ & x & x & 9 & $\alpha_1$ & x & x & 6 & $\alpha_2$ & x & x & 3 & $\alpha_3$ & x & x & 0 & \\
        x & x & $\gamma_0$ & 9 & x & x & $\gamma_1$ & 6 & x & x & $\gamma_2$ & 3 & x & x & $\gamma_3$ & 0 & \\ \hline
        &&&28&&&&$+20$&&&&$+12$&&&&$+4$ & $=64$\\ \hline
    \end{tabular}
    \end{adjustbox}
    \caption{Method 1, each generator is a diagonal element $(\alpha_i,\beta_i,\gamma_i)$. Every vertex is reached by the edges generated.}
    \label{tab:shapes1}
\end{table}

\begin{table}
    \centering
    \begin{adjustbox}{width=\textwidth}   
    \begin{tabular}{||ccc|c||ccc|c||ccc|c||ccc|c||c||} 
        \hline
        \multicolumn{4}{||c||}{Set 1} & \multicolumn{4}{c||}{Set 2} & \multicolumn{4}{c||}{Set 3} & \multicolumn{4}{c||}{Set 4} &  \\
        \hline
        \multicolumn{3}{||c|}{vertices} & \# & \multicolumn{3}{c|}{vertices} & \# & \multicolumn{3}{c|}{vertices} & \# & \multicolumn{3}{c|}{vertices} & \# & Tot \\
        \hline
        $\alpha_0$ & $\beta_0$ & $\gamma_0$ & 1 & $\alpha_0$ & $\beta_0$ & $\gamma_1$ & 0 & $\alpha_0$ & $\beta_0$ & $\gamma_2$ & 0 & $\alpha_0$ & $\beta_0$ & $\gamma_3$ & 0 & \\ 
        $\alpha_0$ & $\beta_0$ & x & 3 & $\alpha_0$ & $\beta_0$ & x & 0 & $\alpha_0$ & $\beta_0$ & x & 0 & $\alpha_0$ & $\beta_0$ & x & 0 & \\
        $\alpha_0$ & x & $\gamma_0$ & 3 & $\alpha_0$ & x & $\gamma_1$ & 0 & $\alpha_0$ & x & $\gamma_2$ & 0 & $\alpha_0$ & x & $\gamma_3$ & 0 & \\
        x & $\beta_0$ & $\gamma_0$ & 3 & x & $\beta_0$ & $\gamma_1$ & 3 & x & $\beta_0$ & $\gamma_2$ & 3 & x & $\beta_0$ & $\gamma_3$ & 3 & \\
        $\alpha_0$ & x & x & 9 & $\alpha_0$ & x & x & 0 & $\alpha_0$ & x & x & 0 & $\alpha_0$ & x & x & 0 & \\
        x & x & $\gamma_0$ & 9 & x & x & $\gamma_1$ & 9 & x & x & $\gamma_2$ & 9 & x & x & $\gamma_3$ & 9 & \\ \hline
        &&&28&&&&$+12$&&&&$+12$&&&&$+12$ & $=64$\\ \hline
    \end{tabular}
    \end{adjustbox}
    \caption{Method 2, the generators (first row) belong to the same row.}
    \label{tab:shapes2}
\end{table}

\begin{table}
    \centering
    \begin{tabular}{||cc|c||cc|c||cc|c||cc|c||c||} 
        \hline
        \multicolumn{3}{||c||}{Set 1} & \multicolumn{3}{c||}{Set 2} & \multicolumn{3}{c||}{Set 3} & \multicolumn{3}{c||}{Set 4} &  \\
        \hline
        \multicolumn{2}{||c|}{vertices} & \# & \multicolumn{2}{c|}{vertices} & \# & \multicolumn{2}{c|}{vertices} & \# & \multicolumn{2}{c|}{vertices} & \# & Tot \\
        \hline
        $\alpha_0$ & $\gamma_0$ & 1 & $\alpha_1$ & $\gamma_1$ & 1 & $\alpha_2$ & $\gamma_2$ & 1 & $\alpha_3$ & $\gamma_3$ & 1 & \\ 
        $\alpha_0$ &  x & 3 & $\alpha_1$ & x & 2 & $\alpha_2$ &  x & 1 & $\alpha_3$ &  x & 0 & \\
        x & $\gamma_0$ & 3 & x & $\gamma_1$ & 2 & x &  $\gamma_2$ & 1 & x & $\gamma_3$ & 0 & \\ \hline
        &&7&&&$+5$&&&$+3$&&&$+1$ & $=16$\\ \hline
    \end{tabular}
    \caption{Method 1, each set is generated from a vertex $(\alpha_i,\gamma_i)$.}
    \label{tab:8d_shapes1}
\end{table}

\begin{table}
    \centering
    \begin{tabular}{||cc|c||cc|c||cc|c||cc|c||c||} 
        \hline
        \multicolumn{3}{||c||}{Set 1} & \multicolumn{3}{c||}{Set 2} & \multicolumn{3}{c||}{Set 3} & \multicolumn{3}{c||}{Set 4} &  \\
        \hline
        \multicolumn{2}{||c|}{vertices} & \# & \multicolumn{2}{c|}{vertices} & \# & \multicolumn{2}{c|}{vertices} & \# & \multicolumn{2}{c|}{vertices} & \# & Tot \\
        \hline
        $\alpha_0$ & $\gamma_0$ & 1 & $\alpha_0$ & $\gamma_1$ & 0 & $\alpha_0$ & $\gamma_2$ & 0 & $\alpha_0$ & $\gamma_3$ & 0 & \\ 
        $\alpha_0$ &  x & 3 & $\alpha_0$ & x & 0 & $\alpha_0$ &  x & 0 & $\alpha_0$ &  x & 0 & \\
        x & $\gamma_0$ & 3 & x & $\gamma_1$ & 3 & x &  $\gamma_2$ & 3 & x & $\gamma_3$ & 3 & \\ \hline
        &&7&&&$+3$&&&$+3$&&&$+3$ & $=16$\\ \hline
    \end{tabular}
    \caption{Method 2, the generators belong to the same row.}
    \label{tab:8d_shapes2}
\end{table}

\end{document}